\documentclass[a4paper,USenglish,thm-restate]{lipics-v2021}
\hideLIPIcs
\title{How local constraints influence network diameter and applications to LCL generalizations}

\author{Nicolas Bousquet}{CNRS, INSA Lyon, UCBL, LIRIS, UMR5205, F-69622 Villeurbanne, France }{}{}{}{}

\author{Laurent Feuilloley}{CNRS, INSA Lyon, UCBL, LIRIS, UMR5205, F-69622 Villeurbanne, France }{}{}{}{}

\author{Théo Pierron}{CNRS, INSA Lyon, UCBL, LIRIS, UMR5205, F-69622 Villeurbanne, France }{}{}{}{}

\authorrunning{N. Bousquet, L. Feuilloley, T. Pierron}

\nolinenumbers{}

\usepackage[utf8]{inputenc}
\usepackage[T1]{fontenc}
\usepackage{needspace}
\usepackage{comment} 

\usepackage{libertine} % text font
\usepackage{inconsolata} % texttt font

\usepackage{amsmath, amsthm, amssymb}
\usepackage{stmaryrd}
\usepackage{graphicx}
\usepackage{tikz}
\usepackage{enumerate}
\usepackage{subcaption}

\usepackage{thmtools}
\usepackage{thm-restate}
\usepackage{float}
\usepackage{xcolor}
\usepackage{doi} % hyperlink doi numbers

\newtheorem{question}{Question}

\DeclareMathOperator{\diam}{\mathrm{diam}}

\newcommand{\midd}{\mathcal{M}}



\begin{document}

\maketitle

\begin{abstract}
In this paper, we investigate how local rules enforced at every node can influence the topology of a network. 
More precisely, we establish several results on the diameter of trees as a function of the number of nodes, as listed below. 
These results have important consequences on the landscape of locally checkable labelings (LCL) on \emph{unbounded} degree graphs, a case in which our lack of knowledge is in striking contrast with that of \emph{bounded degree graphs}, that has been intensively studied recently. 

First, we show that the diameter of a tree can be controlled very precisely by a local checker (that is, a distributed decision algorithm that accepts a graph iff all nodes accept locally), granted that its checkability radius is allowed to be at least $2$ (and that the target diameter is not too close to $n$). 
As a corollary, we prove that the gaps in the landscape of LCLs (in bounded-degree graphs) basically disappear in unbounded degree graphs.

Second, we prove that for checkers at distance 1, the maximum diameter can only be trivial (constant or linear), while the minimum diameter can in addition be $\Theta(\log n)$ and $\Theta(n^{1/k})$ for $k\in \mathbb{N}$. These functions interestingly coincide with the known regimes for LCLs. 

Third, we explore computational restrictions of local checkers. 
In particular, we introduce a class of checkers, that we call \emph{degree-myopic}, that cannot distinguish perfectly the degrees of their neighbors. 
With these checkers, we show that the maximum diameter can only be  $O(1)$, $\Theta(\sqrt{n})$, $\Theta(\log{n}/\log \log n)$, $\Theta(\log{n})$, or $\Omega(n)$. 
Since gaps do appear in the maximum diameter, one can hope that an interesting LCL landscape exists for restricted local checkers. 

In addition to the LCL motivation, we hope that our distributed lenses can help give a new point of view on how global structures, such as living beings, can be maintained by local phenomena; understanding the trade-off between the power of the checking and the possible resulting shapes. 

\end{abstract}

%\newpage{}

%\tableofcontents{}

%\newpage{}

\section{Introduction}

\subsection{Questions and motivations}

A general question in distributed computing is: how well can we control global parameters of a system if we can only check it partially and in a distributed manner?
Here, we are interested in the following graph-oriented version of this question: How do local constraints influence the shape of a network? 
More precisely, if we define a set of rules to be satisfied at every node of the network, what are the networks that satisfy these rules? What properties do they have? And conversely, if we want to ensure a given global property, can we obtain it by only enforcing local constraints?

In this paper, we focus on a concrete version of the question, where the networks are (colored) graphs, and in particular \emph{trees}, and the global parameter studied is the \emph{diameter}. Before we discuss our motivations and how our perspective differs from previous work, let us introduce some vocabulary. A \emph{local checker} is a local algorithm run at every node of the network, that outputs a binary decision accept/reject, only based on a neighborhood at constant distance around itself. A network (or equivalently a graph) is globally accepted if every node locally accepts. For example, a local checker that checks if the number of neighbors of the node at hand is at most $\Delta$ accepts exactly all the graphs of maximum degree at most $\Delta$.

%%%%%%%%%%%%%%%%%%%%%%%%%%%%%%%%%%%%%%%%%%%%%%%%%
\paragraph*{Main motivation: Going beyond bounded-degree for LCLs}
%%%%%%%%%%%%%%%%%%%%%%%%%%%%%%%%%%%%%%%%%%%%%%%%%

The most popular problems in the study of locality, in the sense of the LOCAL model, are called \emph{locally checkable labelings} or LCLs for short, introduced in the 90s~\cite{NaorS93}. 
The main characteristic of these problems is that their outputs can be checked locally. For example, to verify that a coloring is correct, every node simply checks that the color it has received as output is different from the ones of its neighbors. 
There are also three finiteness requirements: the number of possible outputs, of possible inputs, and the maximum degree $\Delta$ must be bounded by a constant (\emph{i.e.} independent of the number of nodes $n$).  
Therefore, an LCL can be described by a finite list of correct neighborhoods. In the case of coloring, this list is made of the following neighborhoods: a node of color $c$ with at most $\Delta$ neighbors, each of them being of a color different from $c$.  

After a decade of intense research, we now have a good understanding of the complexity of computing a solution of an LCL problem in the LOCAL model.
More precisely, we know the landscape of complexities for these problems: roughly, we know what are the functions $f$ for which there exists an LCL whose optimal algorithm has complexity~$f(n, \Delta)$.
For example, we know that there is no LCL whose complexity in the LOCAL model lies between $\omega(\log^*n)$ and $o(\log n)$~\cite{ChangKP19}. 

A natural question is: can we generalize this theory by relaxing the finiteness requirements? 
The case of bounded degree and unbounded number of labels has been explored in~\cite{HasemannHRS16}, with the fractional coloring problem\footnote{There are actually two definitions of fractional coloring, the second one enforcing a bounded number of labels, see e.g.~\cite{BousquetEP21, BalliuKO21}.}. Here we are interested in the other direction: going beyond bounded degree. (This question was very recently tackled by \cite{LievonenPS24} as we discuss later.)

Let us now illustrate why the diameter of the network plays a key role in the context of LCL complexities. 
In the bounded-degree case, two key complexities for deterministic algorithms are $\Theta(n)$ and $\Theta(\log n)$. The problems of complexity $\Theta(n)$ are called \emph{global problems}; they are generally the ones where, in a path, the two endpoints have to coordinate their outputs, which requires them to run for $\Omega(n)$ rounds (and $O(n)$ is enough for any problem in the LOCAL model). The problems of complexity $\Theta(\log n)$ are typically the ones where the hardest instances locally look like complete regular trees, and intuitively each node needs to see a leaf to be able to decide its output. 
Therefore, what dictates the complexity for these two cases is, respectively, the maximum and the minimum diameter of trees of bounded degree, which are respectively order of~$n$ and~$\log n$.
Note that for these observations, restricting to trees is not problematic, since everything holds in trees. 
In general, trees are essential in the LCL theory, which justifies why we focus on these graphs; for example, the key technique of round elimination (see~\cite{Suomela20} and references therein) basically works only in trees.\footnote{Another motivations for restricting to trees is that, under some restrictions, if some graph accepted by our generalization of LCL's contains cycles, then the resulting diameters are trivial, as proved at the end of the paper.}

By understanding the diameters of trees accepted by local checkers that are not simply checking that the degree is below $\Delta$, we want to pave the way of a complexity landscape beyond bounded degree. Hence, the first type of questions we want to investigate is the following.

\begin{question}
    For a given local checker, what is the minimum and maximum diameter of the accepted graphs, as a function of $n$?
\end{question}

Also, since it seems too optimistic to hope for a nice landscape in general graphs, we want to explore which restrictions are worth studying in the future. For example, restrictions for which there are gaps in the possible maximum and minimum diameter.

\begin{question}
    What are natural classes of local checkers such that there are gaps in the landscape of maximum and/or the minimum diameter of graphs accepted by such checkers (\emph{e.g.} diameters that are impossible to obtain)?
\end{question}

Let us now quickly summarize the approach and results of  \cite{LievonenPS24}, on this question of understanding complexity landscape beyond bounded-degree. First, the authors sketch how one can create problems of arbitrary complexity between constant and $\Theta(\log n)$, when the degree is unbounded, by using a problem of complexity $O(\log_{\Delta}n)$ and then ignoring a specific number of adjacent edges at each node to artificially reduce the degree~$\Delta$. Second, they restrict the scope in two ways: (1) they study only binary labeling problems (that consist in selecting edges, like in the matching problem), where (2) the constraints are \emph{structurally simple} (roughly,  in a proper labeling, the number of selected (resp. non-selected) edges around every node has to be either smaller than a constant, or polynomially large in the degree). In this setting, they can characterize precisely what happens in the logarithmic regime, thanks to a degree-aware variant of the rake-and-compress technique. 
Our approach is quite orthogonal, as we study the structure of the graphs and do not discuss what are the problems that are solved or the generic algorithms that could be used. Nevertheless, the general story has a similar flavor: we will show in a different, more general and precise way how arbitrary complexities can be obtained, and then restrict to a setting where gaps reappear.

%%%%%%%%%%%%%%%%%%%%%%%%%%%%%%%%%%%%%%%%%%%
\paragraph*{Second motivation: Maintaining global structure locally}
%%%%%%%%%%%%%%%%%%%%%%%%%%%%%%%%%%%%%%%%%%%

Our second motivation is a more exploratory one. We would like to propose a new perspective on how structures can be maintained without centralized monitoring. In other words, we want to explore how simple local rules on small-scale entities shape larger-scale objects. One can think of a crude modeling of living beings, or of self-organizing swarm of robots.

From this perspective, our motivation for studying the diameter is that it is maybe the most basic characteristic of a shape, and the one for studying trees is that these are simple shape that appear in Nature. 

Now, instead of starting from a checker and analyzing the possible diameters, we want to start from some target diameter (as a function of $n$), e.g. one that would have benefits for the global entity, and ask whether we can have a local checker that maintains this diameter.

\begin{question}\label{q:f-given}
    Given some function $f$, can we design a local checker such that the accepted $n$-vertex graphs have diameter $f(n)$?
\end{question}

Also, since it is unrealistic in Nature or in robots to have unlimited local computation, we refine Question~\ref{q:f-given}, by considering the complexity of the local checking.

\begin{question}
    Given some function $f$, can we design a local checker with limited computational capabilities, such that the accepted $n$-vertex graphs have diameter $f(n)$?
\end{question}

The rest of the introduction is organized as follows: in Subsection~\ref{subsec:model} we introduce the precise definitions, in Subsections~\ref{subsec:results-general-checkers} and~\ref{subsec:results-restricted-checkers} we review our results and techniques for general and restricted checkers, respectively.  In Subsection~\ref{app:further-related-work}, we review additional related work.

%%%%%%%%%%%%%%%%%%%%%%%%%%%%%%%
\subsection{Model and definitions}
\label{subsec:model}
%%%%%%%%%%%%%%%%%%%%%%%%%%%%%%%

The graphs/trees considered in this paper are simple and loopless. They can be vertex colored, with a constant number of colors. Remember that the diameter of a graph is the length of the largest shortest path between two nodes, in terms of number of edges. 
%For simplicity, we will say that the number of colors is a property of the checker, and the graphs it works on will implicitly have that number of color (at most).
Inspired by the definition of LCLs and by our biological motivation, we consider anonymous networks, \emph{e.g.} the nodes do not have identifiers. 

\begin{definition}[View of a node]
    The \emph{view} at distance $d$ of a node $v$ in a graph $G$ is the subgraph of $G$, that contains all nodes at distance at most $d$ from $v$, and all the edges with at least one endpoint at distance at most $d-1$ from $v$, and where $v$ is marked as the center. 
\end{definition}

\begin{definition}[Local checker ; $\mathcal{L}_{c,d}$ ; checkability radius]
    A local checker at distance $d$ and with $c$ colors is a local algorithm that is used on a $c$-colored graph, such that, applied to a vertex $v$, it takes the neighborhood at distance $d$ around $v$ and chooses an output: accept or reject. We denote by $\mathcal{L}_{c,d}$ the set of all local checkers at distance $d$ with $c$ colors. The distance $d$ is called the \emph{checkability radius}.
\end{definition}

For short, we sometimes simply use \emph{checker} instead of \emph{local checker}. The definition of a degree-myopic local checker is given in Subsection~\ref{subsec:results-restricted-checkers}, along with the discussion of its origin. 

\begin{definition}[Class accepted/recognized by a local checker]
    Given a local checker $L$, the class of (colored) trees \emph{accepted} (or equivalentaly \emph{recognized}) by this checker, denoted by $\mathcal{C}(L)$, is the set of trees such that on every node the checker $L$ accepts. 
\end{definition}

\begin{definition}[Generalized-LCL]
    Let a \emph{generalized-LCL} be a problem where each node has to choose an output from a finite set, such that the correct output configurations can be recognized by a local checker.
\end{definition}

Intuitively, a checker has maximum (resp. minimum) diameter $f(n)$ if all trees recognized by this checker have diameter at most (resp. at least) $f(n)$, when $n$ is the number of nodes, and this bound is tight. The proper definitions given below, use infimum/supremum instead of maximum/minimum, but we keep the names maximum/minimum, since they are more intuitive.

\begin{definition}[Exact/minimum/maximum diameter of a checker]
    Let $L$ be a local checker. 
    \begin{itemize}
        \item $L$ has \emph{exact diameter} $f$ if for all $n$, $\forall G\in \mathcal{C}(L), |V(G)|=n, Diam(G)=f(n)$.
        \item $L$ has \emph{maximum diameter} $f$ if, for all $n$, $\sup_{k \ge n}\{Diam(G)/f(k)$ for all $G$ such that $|V(G)|=k$ and $G \in L\}= 1$.
        \item $L$ has \emph{minimum diameter} $f$ if, for all $n$, $\inf_{k \ge n}\{Diam(G)/f(k)$ for all $G$ such that $|V(G)|=k$ and $G \in L\}= 1$.     
    \end{itemize}
\end{definition}

%%%%%%%%%%%%%%%%%%%%%%%%%%%
\subsection{Discussion, results and techniques for general checkers}
\label{subsec:results-general-checkers}
%%%%%%%%%%%%%%%%%%%%%%%%%%%

Let us first review our results about local checkers without restriction on the computation power of the nodes.

\paragraph*{Warm up for maximum and exact diameter}

Let us start by introducing some basic intuitions for maximum and exact diameter.
First, let us consider a local checker $L$ at distance $1$, without colors. Such a checker is basically of the following form: if the degree of the node belongs to some specific set of integers $S$, then accept, otherwise reject. Necessarily, $1$ belongs to $S$ because we need to allow leaves, if we want to accept finite trees.
Let us informally prove that the maximum diameter of $L$ is $\Omega(n)$. Consider a tree $T$ accepted by~$L$, and two of its edges $uv$ and $wz$. The deletion of these edges leaves three connected components: $T_{middle}$ which is the part of the tree between $uv$ and $wz$, and $T_u$ and $T_{z}$, the connected components of $u$ and $z$, respectively. Then, we can define new trees by replacing $T_{middle}$ by an arbitrary number of copies of $T_{middle}$ organized into a chain (identifying copies of $wz$ in one copy with $uv$ in the next one). The  set of degrees in the new tree is the same
as in the original tree. Hence, all these new trees are also accepted, and they have asymptotic linear diameter. 
The key point is that we can use the fact that the neighborhoods were indistinguishable to make the graph path-like. We refer to this "pumping" technique as \emph{grafting} and will be formally defined in full generality in Section~\ref{sec:maxdiam} for any possible distance and number of colors.

Now, for a positive result, consider a checker at distance $2$ (again without colors). It basically manipulates the degree of its neighbors. Then it is possible to have a checker that recognizes exactly the graphs of the following form: take path of length $i$, choose one endpoint $u$, and attach $i$ leaves to the node of the path at distance $i$ from $u$. Indeed, the nodes can check that they are either leaves or that their non-leaf neighbors follow the increasing degree sequence. This leads to a local checker that accept trees of maximum diameter $\Omega(\sqrt{n})$.  

An important aspect of the previous argument is that all nodes had different degrees. This actually works with a weaker condition, namely, there is no repeated pair of degrees of adjacent vertices. Indeed, if a pair is repeated, then we can use pumping to prove that the checker accepts trees of linear diameter. By using a longer degree sequence without such repetitions (inspired by de Bruijn sequences), one can actually improve the previous bound to $n^{2/3}$. Moreover, if we want to target a diameter that is much larger than $n^{2/3}$, we would necessarily get some repeated chunks, yielding linear diameter. In some sense, carrying the information about the distance from an extremity uses some quantity of nodes, and at some point we run out of nodes. If we want to go further, then there is some repetition and the maximum diameter moves to the boring linear regime.

But one can also wonder if we can obtain smaller diameters. Suppose that we want to target a diameter that is smaller than $n^{2/3}$. One can show that with the same kind of construction, except that we truncate the main path and attach enough leaves to the last node to get exactly the diameter we want, we can obtain any function of~$n$ which is $O(n^{2/3})$ as a diameter. This can be recognized by a local checker, because the last node can check that the number of leaves it has is consistent with the degree of the second-to-last node, which measures the diameter. We refer to this technique as \emph{padding}.  

Our results on maximum and exact diameters, stated in the next paragraph, are generalizations of these constructions.
%\textcolor{blue}{Our results on maximum and exact diameters are refinements of these first constructions. We prove that we can carry information about the distance to an extremity with few nodes additional nodes via fancy encodings, we use padding to get the exact diameter we want, and we prove that when we get too close to $n$, then we cannot keep the diameter on a leash, as it jumps to linear.}

%We first review our results on the maximum and exact diameter and their corollaries for generalized LCLs, and then discuss the results about minimum diameter. 

\paragraph*{Formal results and more techniques on maximum and exact diameter}

Let us now review our formal results.
First, we prove that, for each possible distance $d$ and every number of colors, there is a threshold function $S_{c,d}(n)$ such that local checkers accept trees of linear diameter or trees of diameter at most $O(S_{c,d}(n))$. In other words, not all the maximum diameters can be obtained with local checkers. More formally, we prove in Section~\ref{sec:maxdiam} that the following holds:

\begin{restatable}{theorem}{ThmGapResult}
\label{thm:gap_result}
Every local checker in $\mathcal{L}_{c,d}$ has maximum diameter at most $(4d^2+4d+1)\cdot S_{c,d}(n)$ or $\Theta(n)$, where $S_{c,1}(n)= c^2/9$, $S_{c,2}(n)=(cn^c)^{2/(2c+1)}$, $S_{c,3}(n)=36n/\log^2n$ and $S_{c,d}(n)=4n/g_d(\log n)$ if $d>3$.
\end{restatable}

For large values of $d$, the threshold function depends on the inverse $g_d$ of the function $x\mapsto x/\log^{(d-3)}x$, where the log is iterated $d-3$ times. Note that one can think of the growth of $g_d$ as slightly super-linear, so the last case of the following theorem is roughly $\Theta(n/\log n)$. 

\begin{comment}
\textcolor{red}{N. En fait vu le warm up que Laurent a rajouté je sais pas si ça sert de garder cette partie... La seule partie que je trouve pertinente est la partie sur de Bruijn qu'on peut peut être aussi mettre dans le warm up?}
\textcolor{purple}{Theorem~\ref{thm:gap_result} is tight and provides both upper bounds and tight constructions. Let us briefly give the main ingredients of the proof. For the upper bound, the proof makes use of the following remark: if the diameter of an accepted tree is $D$, then most of the vertices in a path $P$ of length $D$ cannot see more than $O(n/D)$ vertices in their balls of radius $d$. In order to prove the gap between $S_{c,d}(n)$ and $\Theta(n)$, we prove that, if $D$ is larger than $S_{c,d}(n)$ then two vertices of $P$ must have exactly the same view. When two such vertices exist, we can prove that what we call a \emph{grafting argument} (similar to the one sketched earlier) permits to increase the diameter by additional constant while adding a constant (possibly extremely large) number of vertices. The repetition of the grafting argument permits to obtain arbitrarily large trees with linear diameters.
}
\end{comment}

We moreover provide constructions reaching these upper bounds. To do so, we generalize the idea of paths with a prescribed degree sequence by defining a class of checkers that accept only trees of specific shapes. This allows to get a correspondence between the accepted trees and the strings over a given alphabet. In that correspondence, each checker is associated with a language, where the membership of a string in the language depends only on bounded-size infixes. We then use this machinery on languages to provide tight examples for Theorem~\ref{thm:gap_result}.

%We complete Theorem~\ref{thm:gap_result} with another result which ensures that the upper bounds are tight. This second result first ensures that there exist local checkers in $\mathcal{L}_{c,d}$ whose diameter match the bounds of Theorem~\ref{thm:gap_result} which ensures its tightness. 
We finally provide a generic construction that ensures that it is possible to build a local checker that accept trees whose diameters is precisely any function which is a $O(S_{c,d}(n))$. In other words, there is no gap in the set of diameters that can be recognized below $S_{c,d}(n)$ by local checkers since all the possible diameters can be recognized.  More formally, we prove in Section~\ref{sec:prescribed} that the following holds:

\begin{restatable}{theorem}{ThmNoGap}
\label{thm:no_gap}
For every function $f(n)=O(S_{c,d}(n))$, there exists a local checker in $\mathcal{L}_{c,d}$ of exact diameter $f(n)$.
\end{restatable} 

Note that Theorems~\ref{thm:gap_result} and~\ref{thm:no_gap} give fine-grained versions of the intuitions given in the warm-up for distance 1 and 2. For distance 1, because of the colors, pumping works above some constant depending on the colors, while for distance 2, the threshold is polynomial in $n$. Then for larger distance we get nearly linear thresholds, which means that we can target almost any diameter.

From this theorem, we can derive the following corollary, which proves that the landscape for generalized-LCLs (as defined earlier) has basically no gaps. 

\begin{restatable}{corollary}{CoroLCL}
\label{coro:LCL}
    For every function $f(n)=O(S_{c,d}(n))$, there exists a generalized-LCL at distance $d+4$ of complexity~$f(n)$. 
\end{restatable}

We prove this corollary in Appendix~\ref{app:coroLCL}. The main idea consists in mixing the checkers of Theorem~\ref{thm:no_gap} with a global problem, namely $2$-coloring, forcing the complexity to be of the same order as the  diameter. This is more complicated than it looks at first sight, because the very specific constructions we use for the theorem are such that the nodes can know exactly where they are in the path. In particular, they know the parity of their distance to the first node of the path, hence 2-coloring is \emph{not} a global problem in these graphs. We modify slightly the construction by subdividing some edges to introduce some uncertainty, which fixes this issue (which is possible only by increasing a bit the checking radius and with a precise understanding on the local checkers of Theorem~\ref{thm:no_gap}).

To finish this part on maximum diameter, let us mention that in Appendix~\ref{sec:cycles}, we show that if we were to allow cycles, then at distance 1, the behavior would be the same as for trees: the maximum diameter is either constant or linear.

\paragraph*{Results on minimum diameter}

We now turn our attention to minimum diameter.
Note that all the results obtained for exact diameter also hold in the case of minimum diameter. In particular, Theorem~\ref{thm:no_gap} ensures that for every $d\ge 2$, we can construct a local checker accepting trees whose minimum diameter is almost everything. Therefore, these results leave as open the behavior of the minimum diameter in the range $\omega(S_{c,d}(n))$ to $o(n)$. We devote Section~\ref{sec:min-diameter} to study the particular case $d=1$ where the dichotomy we obtained for maximum diameter leaves a wide range of possibilities for minimum diameter. In this case, we show that the landscape is quite different since the minimum diameter is either constant, logarithmic or $\Theta(n^{1/k})$ for some values of $k$ which are at most $c^2$.

\begin{restatable}{theorem}{ThmRadiusOneMin}
\label{thm:radius-1-MIN}
Let $c$ be an integer. The minimum diameter of any $L\in\mathcal{L}_{c,1}$ is either constant, logarithmic or $\Theta(n^{1/k})$ for some values of $k\leqslant c^2$.
\end{restatable}

We note that these are well-known regimes for (bounded-degree) LCLs. 
In particular,  the $n^{1/k}$ regime has been identified precisely in~\cite{ChangP19, Chang20} and the graphs at the core of both their proof and our proof are the same. For example, for $k=2$, these are the graphs made of a path of length $\sqrt{n}$, where we attach a path of length $\sqrt{n}$ to every node. For $k=3$, we take the same construction but with $n^{1/3}$ instead of $n^{1/2}$, and add one more level to the recursion etc.
In a nutshell, the graphs of the construction are recursive rake-like graphs, where all the levels are balanced.
Now the two proofs are pretty different. In ~\cite{ChangP19, Chang20} the arguments are algorithmic: they show that via rake-and-compress, that every tree can be framed as such a rake-like graph, and that the hardest instances are the balanced ones. In our proof, the spirit is that there is some underlying hierarchy of colors, that leads to the hierarchical structure of the graph, and in order to get a minimum diameter we want to balance all the levels, and we show this is always possible, by doing some subtle pumping argument (based on a non-trivial partial order on the pair of colors).

Intriguingly, we do not know whether the results on LCLs and our minimum diameter result are related, in the sense that one would imply the other. Clarifying this relation would be very interesting.

%%%%%%%%%%%%%%%%%%%%%%%%%%%
\subsection{Discussion, results, and techniques on restricted checkers}
\label{subsec:results-restricted-checkers}
%%%%%%%%%%%%%%%%%%%%%%%%%%%

From the viewpoint of our original motivations, the results we get for general checkers are not very satisfying: they give a pessimistic answer for LCL generalization (no hope to find gaps in the general unbounded-degree case) and the constructions are too contrived to give insights about natural structures, such as living beings. 
Therefore, we want to consider a restricted model, that will hopefully reveal some gaps for the maximum diameter (hence avoid statements such as Corollary~\ref{coro:LCL}) and forbid unnatural constructions. 

We start by discussing what conditions such a restricted model should satisfy and propose our model, \emph{degree-myopic checkers}. Then we discuss the maximum diameter landscape for these checkers, and our proof techniques. 

\paragraph*{Discussions of the flaws of general checkers}

Three elements in our general checker construction feel unnatural.  
First, in many proofs we use padding: the last node of a path-like structure can observe the neighborhood of the preceding node to deduce the diameter-size ratio of the rest of the graph, and then check that it has been given the right number of leaves to modify this ratio in the right way. One expects that in a more natural construction, the graph is more homogeneous, and no node should compute an explicit complicated function. 
Second, we map integers to neighborhoods in some arbitrary way, using a generalization of De Bruijn words. To get the full power of this construction, we need that adjacent nodes can have arbitrarily different degrees, and that they can identify these degree perfectly. In some sense, the nodes should be able to manipulate arbitrarily large and different numbers. 
Third, also needed for the integer to neighborhood mapping: the checking of a node $v$ depends on the full view of the nodes, and cannot be decomposed as a series of simple checks between adjacent nodes. In some sense, we are abusing the fact that we allowed constant checkability radius, instead of the radius 1. 

\paragraph*{Degree-myopic checkers}

We propose a restricted setting for checkers, called \emph{degree-myopic}, which boils down to two ideas.

First, we enforce that a node cannot distinguish perfectly between all its possible neighborhoods, by making some kind of equivalence classes. More precisely, since we will use checkers at distance two (without colors), the crux is that a node is not able to distinguish arbitrarily many possible degrees for its neighbors.
We make this appear concretely by saying that a node will first put its neighbors' degree in a few categories, and later will only be allowed to use these categories, not the actual degree.
Here there is some freedom in how we choose these categories. 
We decide to consider that a node can only manage degrees that are close to its degree and leaves, hence the name of \emph{degree-myopic}. 
Other meaningful variants could be considered.

Formally, let $d$ be the degree of the node at hand, and $d'$ the degree of the neighbor considered. The \emph{types} of the neighbors are:

\begin{itemize}
    \item \textbf{Leaf}: $d'=1$
    \item \textbf{Degree-1}: $d'=d-1$
    \item \textbf{Equal-degree}: $d'=d$
    \item \textbf{Degree+1}: $d'=d+1$
\end{itemize}

Hence, at that point, the knowledge that a node has of its neighborhood is a 4-tuple of integers: how many neighbors are of each type. Note that the degree difference between neighbors is very limited: graphs with two non-leaf neighbors of very different degree will be rejected from the start. 

Second, we decide to restrict the model further, to avoid complex manipulations of the integers of these 4-tuples. For example, we want to avoid checking that the number of leaves is related to the number of neighbors of equal degree by an arbitrarily complicated function.
More precisely, the non-leaf nodes checks that a property $C$ of the following form is satisfied: for every type $i$, we are given two quantities $a_i\in \mathbb{N}$ and $b_i \in \mathbb{N} \cup \{\infty\}$ such that $a_i\leq b_i$, and the number of neighbors of type $i$ is in $[a_i,b_i]$.
To avoid issues around the leaves of the tree, the nodes of degree $a_{Degree-1}$ or less are not required to have at least $a_{Degree-1}$ neighbors of type \emph{Degree-1}. 
Beyond the technical reasons, this is justified by the fact that we are interested in what happens in the core of the trees, where the degrees will be large.

By construction, this model avoids the three issues mentioned earlier. Indeed, thanks to the type equivalence, we cannot use padding, only similar degrees can be compared, and thanks to the shapes of the property $C$, the checking is very constrained. (We will see that it is not useful to have arbitrarily large $a_i, b_i$, hence there is no tricks on this side either.)   

\paragraph*{Landscape result and technique}

Our main result about degree-myopic checkers is the following.

\begin{restatable}{theorem}{ThmClassificationMyopic}
    \label{thm:classification-myopic}
    The possible maximum diameter function for degree-myopic local checkers are $O(1)$, $\Theta(\sqrt{n})$, $\Theta(\log{n}/\log \log n)$, $\Theta(\log{n})$ or $\Omega(n)$. 
\end{restatable}

This result is nice from the viewpoint of our LCL motivation: it shows that for restricted checkers, we do have an interesting landscape. In particular, we get $\Theta(\log n/\log \log n)$ which does not appear in the LCL landscape for trees, since there is a gap between $O(\log^* n)$ and $\Omega(\log n)$~\cite{ChangKP19}. This is also exciting from the viewpoint of our second motivation: we get that simple maximum diameter functions can be obtained by relatively natural checkers. 

The proof (deferred to Appendix~\ref{sec:myopic}) consists in first showing that if a degree-myopic checker does not accept trees of linear diameter, then the accepted trees must be very well-behaved: roughly speaking, we can root the tree, in such a way that all paths from leaves to root are monotone in terms of degrees. Then we do a short case analysis, where it appears that basically the non-trivial extremal behaviors for trees can only be of the following types: 
\begin{itemize}
    \item A caterpillar, where the $i$-th node of the path has degree $i$, which leads to $\Theta(\sqrt{n})$.
    \item A complete $k$-ary tree (for constant $k$), with some pending leaves, which leads to $\Theta(\log n)$. 
    \item A tree, where if the depth is $d$, then the nodes at depth $i$ have degree $d-i+1$, which leads to $\Theta(\log n/\log \log n)$. 
\end{itemize}

\paragraph*{Further related work}

Relating local and global structures of graphs is obviously not a new topic. In Appendix~\ref{app:further-related-work}, we review various research areas touching on this topic (network science, graph theory, and distributed computing), and highlight the differences and similarities with our perspective.

\section{Gap results for maximum diameter}
\label{sec:maxdiam}

The goal of this section is to prove that Theorem~\ref{thm:gap_result} holds. 

%%%%%%%%%%%%%%%%%%%%%%%
\ThmGapResult*
%%%%%%%%%%%%%%%%%%%%%%%

Intuitively, we show that if a checker accepts a tree containing a sufficiently large path, we can find very similar nodes, and fool them by ``pumping'' on the path to obtain more trees that are still accepted by the checker, but with various diameters. This pumping relies on an operation called \emph{grafting} (illustrated on Figure~\ref{fig:grafting}). 

\begin{figure}[!ht]
    \centering
    \includegraphics[width=0.7\linewidth]{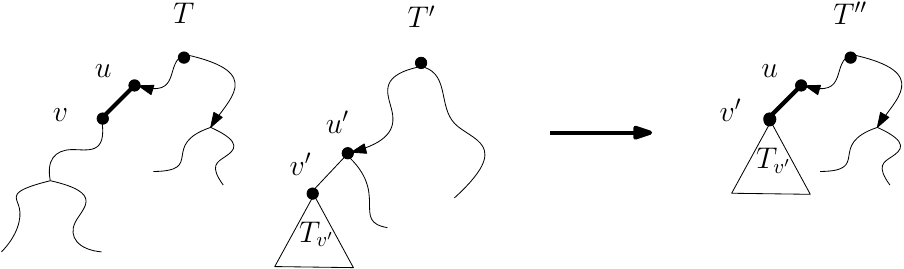}
    \caption[scale=0.7]{The tree $T''$ is the \emph{graft} of $T'$ in $T$ at $(uv,u'v')$.}
    \label{fig:grafting}
\end{figure}

\begin{definition}[Grafting]
\label{def:grafting}
Let $T$ be a (colored) tree and $uv$ be an edge of $T$. We denote by $T_u^{uv}$ (or $T_u$ when $u$ is clear from context) the connected component of $u$ in $T \setminus uv$.
Consider two trees $T$ and $T'$, with edges $uv$ and $u'v'$ respectively. The tree $T''$ which is the \emph{graft} of $T'$ in $T$ at $(uv,u'v')$ is the tree obtained from $T$ by replacing the $T_v$ by $T_v'$. In other words, we remove all the vertices of $T_v$ and replace them by the vertices of $T_v'$ and the edge $uv'$.
\end{definition}

Note that, given a tree $T$ and two edges $uv, u'v'$ such that $u'v'$ belongs to $T_v^{uv}$, we can graft $T$ into itself in two ways: we can graft $T$ in $T$ at $(uv,u'v')$ or in $(u'v',uv)$. Note that both operations are not the same: one of them will reduce the size and the diameter of the tree while the other will increase them.

%Partition $T$ into two subtrees $T_u$ and $T_v$, each containing the nodes that are closer to $u$ or to $v$. Similarly, define $T'_{u'},T'_{v'}$. Denote by $T''$ the tree obtained by replacing $T_v$ by $T'_{v'}$ in $T$. We say that $T''$ is obtained by grafting $T'$ in $T$ at $(uv,u'v')$. 

Let $d\ge 2$. We first show that the grafting operation preserves acceptance by a local checker provided that the graft happens on similar edges. More precisely, define the \emph{view} at distance $d$ from an edge $uv$ in a graph $G$ as the subgraph of $G$, that contains all nodes at distance at most $d-1$ from at least one of the two vertices $u$ or $v$ and all the edges with at least one endpoint at distance at most $d-2$ from $u$ or $v$, and where $uv$ is marked as the center. One can easily check that views are not modified when grafting on edges with the same view. More formally, we prove the following in appendix.

\begin{lemma}\label{lem:grafting}
    Let $L\in \mathcal{L}_{c,d}$ accepting two trees $T,T'$, each containing an edge $uv\in E(T)$, $u'v'\in E(T)$ with the same view at distance $d$. Then $L$ also accepts the graft of $T'$ in $T$ at $(uv,u'v')$. 
\end{lemma}

If a local checker accepts a tree $T$ containing two edges with the same view at distance $d$, one can then graft $T$ into itself at these edges and get a new tree, still accepted by $L$, and with a third edge with the same view. Iterating this argument yields the following (the formal proof being postponed to the appendix).

\begin{lemma}
    \label{lem:pumping_lemma}
    Let $c,d$ be integers and $L$ be a local checker in $\mathcal{L}_{c,d}$. If $L$ accepts a $c$-colored tree $T$ that contains a path going through two edges $uv,xy$ (in this order) such that $uv$ and $xy$ have the same view at distance $d$. Then, $L$ has linear maximum diameter.
\end{lemma}

Observe that at distance $1$, the view of an edge is intuitively just the color of its endpoints. In other words, given an edge $uv$, the color of the neighbors of $v$ has no impact on the fact that $u$ accepts or not. So if we can find a path with two edges colored alike, we will be able to increase the diameter via grafting using Lemma~\ref{lem:pumping_lemma}.
%Therefore, in that case $S_{c,1}(n)=c^2+1$ and Theorem~\ref{thm:gap_result} can be derived from Lemma~\ref{lem:pumping_lemma} by a simple pigeonhole argument that we explicit for the sake of completeness.

\begin{corollary}
\label{cor:radius-1-MAX}
Let $c$ be an integer. The maximum diameter of any $L\in\mathcal{L}_{c,1}$ is either at most $c^2$ or linear.
\end{corollary}

\begin{proof}
Let $L\in\mathcal{L}_{c,1}$ be a local checker accepting a tree $T$ whose diameter is larger than $c^2$. Let $P$ be a path shortest path of $T$ of length $c^2+1$. By pigeonhole principle, $P$ has two edges colored similarly and we can apply Lemma~\ref{lem:pumping_lemma} to conclude.
%And let us denote by $P_i$ the suffix of $P$ starting at its last vertex. Let us denote by $a_1b_1$ the pair of colors of the first edge of $P$. If the pair $ab$ appears again in $P_{2}$, the conclusion follows by Lemma~\ref{lem:pumping_lemma}. So we can assume that $P_{2}$ does not contain any edge colored $ab$. So the first edge of $P_2$ is colored with a pair of colors $a_2b_2$ distinct from $a_1b_1$. Since the number of pairs is at most $c^2$ we should end up on a case where we can apply.
\end{proof}

At distance at least $2$, the view of each edge consists in a pair of $c$-colored rooted trees of height $d$. These may contain arbitrarily many vertices, leading to an unbounded number of possible views. To overcome this issue, we adapt the proof technique by first finding a lot of edges (on a path) whose views contain only few vertices. Denote by $t(d,k)$ the number of $k$-vertices trees of height $d$.

\begin{theorem}[\cite{pach2013number}]
\label{thm:equiv}
The following holds:
\begin{itemize}
    \item $\log t(3,k)\sim \pi\sqrt{2k/3}$
    \item for $d>3$, $\log t(d,k) \sim \frac{\pi^2}{6}\cdot \frac{k}{\log^{(d-3)}(k)}$.
\end{itemize}
\end{theorem}

We are now ready to conclude the proof of Theorem~\ref{thm:gap_result}.

\begin{proof}[Proof of Theorem~\ref{thm:gap_result}]
Let $L\in\mathcal{L}_{c,d}$ and $T$ be a $c$-colored $n$-vertex tree accepted by $L$ with diameter more than $(4d^2+4d+1)\cdot S_{c,d}(n)$. Let $P$ be a path of length more than $(4d^2+4d+1) \cdot S_{c,d}(n)$ in $T$. 

Consider the trees obtained when removing the edges of $P$. Less than $2d\cdot S_{c,d}(n)$ of them contain at least $n/2d\cdot S_{c,d}(n)$ vertices. Moreover, each such tree intersects the view of at most $2d+2$ edges of $P$. In particular, $P$ contains more than $S_{c,d}(n)$ edges whose view consists of $2d$ trees hanging from $P$, all of size less than $n/2d\cdot S_{c,d}(n)$. In particular, their view at distance $d$ contain at most $n/S_{c,d}(n)$ vertices. 

When $d=2$, there are at most $c$ choices for the color of each vertex $v$, and at most $\deg^c(v)$ choices for the colors of its neighbors. In particular, we found more than $S_{c,2}(n)$ edges on $P$ that can have $c^2\cdot n^{2c}/S_{c,2}(n)^{2c}=S_{c,2}(n)$ different views. We can thus find two edges with the same view and conclude using Lemma~\ref{lem:pumping_lemma}.

When $d>2$, the view of each endpoint of our edges is a tree of height $d$ with $k:=n/S_{c,d}(n)$ nodes, whose $k$ nodes are colored with $c$ colors, hence there are at most $k^c\cdot t(d,k)$ such views. Therefore, the number of views of our edges is at most $k^{2c}\cdot t(d,k)^2$. Using Theorem~\ref{thm:equiv}, we get (for large enough $k$) that $k^{2c+1}\cdot t(3,k)^2\leqslant e^{6\sqrt{k}} \leqslant n$ since $k=\log^2 n/36$ in that case. Similarly, for $d>3$, we get $k^{2c+1}\cdot t(d,k)^2\leqslant n$ since $k=g_d(\log n)/4$ in that case. Therefore, in both cases, we get $k^{2c}\cdot t(d,k)^2\leqslant n/k = S_{c,d}(n)$, and we can again conclude using Lemma~\ref{lem:pumping_lemma}.
\end{proof}

%%%%%%%%%%%%%%%%%%%%%%%%%%%%%%%%%%%%%%%%%%%%%%%%%%%%%%%%%
%%%%%%%%%%%%%%%%%%%%%%%%%%%%%%%%%%%%%%%%%%%%%%%%%%%%%%%%%
%%%%%%%%%%%%%%%%%%%%%%%%%%%%%%%%%%%%%%%%%%%%%%%%%%%%%%%%%

\section{Construction of local checkers of prescribed exact diameter for $d \ge 2$}\label{sec:prescribed}

The goal of this section is to prove Theorem~\ref{thm:no_gap}.

%%%%%%%%%%%%%%%%%%%%%%%%%%%
\ThmNoGap*
%%%%%%%%%%%%%%%%%%%%%%%%%%%

The case of distance $1$ is very easy: one can easily construct a local checker in $\mathcal{L}_{2,1}$ recognizing stars, which have constant exact diameter.
So in the rest of the section, we focus on the case $d\ge 2$. We will actually prove that the statement holds even restricted to a class of local checkers that will allow to associate the trees they recognize with strings over some alphabet.
Then we will provide local checkers whose maximum diameter is precisely $S_{c,d}(n)$. We will finally explain how we can adapt our construction using a general tool to accept only trees of diameter $f(n)$ for any function  $f(n)=O(S_{c,d}(n))$.

%%%%%%%%%%%%%%%%%%%%%%%%%%%%%%%%%%
\subsection{Encoding sequences and locally testable languages.} 
%%%%%%%%%%%%%%%%%%%%%%%%%%%%%%%%%%

A \emph{caterpillar} is a tree $T$ such that the removal of the leaves of $T$ yields a path $P$, called \emph{the backbone} of $T$. A \emph{$d$-caterpillar} is a tree $T$ such that if we iteratively remove all the leaves $d$ times then the remaining graph is a path. Equivalently, all the vertices of $T$ are at distance at most $d$ from $P$. Note that a caterpillar is a $1$-caterpillar.

Let us now explain how we can encode a word (with an infinite alphabet) into a $d$-caterpillar and conversely. Let $\Sigma=(a_n)_{n\in\mathbb{N}}$ be an infinite alphabet together with an ordering on the letters.
Let $d$ be an integer and $f$ be a bijection associating a colored rooted tree of depth at most $d$ to each letter. Given a string $s=s_1\cdots s_p$ over the alphabet $\Sigma$, the tree $T[s]$ is obtained by taking the trees $f(s_1),\ldots,f(s_p)$, adding an edge between the root of $f(s_i)$ and the root of $f(s_{i+1})$ for $i\in[0,p]$, where by convention $f(s_0)$ and $f(s_{p+1})$ are paths of length $d+1$ (rooted at a leaf). Note then that $T[s]$ is a $d$-caterpillar whose backbone is formed by the roots of $f(s_0)\ldots,f(s_{p+1})$. For every vertex of the backbone, we say that $a_i$ is \emph{associated} to it if $f(a_i)$ is attached on it. Also note that every $d$-caterpillar can be seen as $T[s]$ for some choice of $s$ since $f$ is bijective.

In order to fully control the word, we would like to enforce that we read it in the right direction. And with the caterpillar defined above, both $s$ and its mirror provide the same caterpillar. Moreover, vertices of the middle of backbone of $T$ have a priori no reason to know in which direction they are supposed to "read" the word. To enforce one direction of reading, we actually slightly change the definition of $T[s]$: instead of creating a copy of $f(s_i)$ for each $s_i$, we create three of them, identify their roots, and add $i\mod 3$ pending leaves to the root. (Note that this only changes the number of vertices by a constant multiplicative factor, and the diameter is not affected since $d \ge 1$). Now, the numbers of neighbors outside the backbone of vertices in the backbone mod $3$ yields the sequence $0,1,2,0,1,2,\ldots$. 
We say that a caterpillar is \emph{special} if this property is satisfied and paths of length $d+1$ are attached to the first and last vertices of the backbone. In particular, since $f$ is bijective, each special caterpillar can be uniquely written as $T[s]$ for some string $s$ over $A$.

We are interested in local checkers, called \emph{special checkers} (or $d$-special checkers), that only accept special $d$-caterpillars. These can be easily enforced as shown by the following lemma, whose proof is deferred to the appendix.

\begin{lemma}\label{lem:d-caterpillar}
    A local checker at distance $d+1$ can check:
    \begin{itemize}
        \item if a tree is a special $d$-caterpillar and,
        \item every vertex can determine if it is on the backbone or not as well as its neighbors on the backbone and,
        \item every vertex can determine if it is an endpoint of the backbone and,
        \item every vertex can determine the letter associated to it as well as the letter associated to its neighbors on the backbone. 
    \end{itemize}
\end{lemma}

We denote by $\mathcal{L}_{c,d+1}^*$ the set of local checkers that only accept special $d$-caterpillars.
We may now associate with each $L\in\mathcal{L}_{c,d+1}^*$ the language of strings $s$ such that $T[s]$ is accepted by $L$. Observe that the membership of a string in such languages depends only on its set of substrings of size at most $2d-1$. Such languages are said to be \emph{locally testable} (or \emph{$d$-testable}). In particular, for every possible encoding $f$ and every locally testable language, there is a local checker in $\mathcal{L}_{c,d}^*$ accepting exactly the trees $T[s]$ for $s\in L$.  In the rest of this section we will need two simple locally testable languages. Let $\Sigma=(a_n)_{n \in \mathbb{N}}$ be a (possibly infinite) alphabet where letters are ordered.

Then one can easily check that the following language is locally testable for every $d\ge 2$ (since a substring of $2d-1$ characters has to be a consecutive sequence of letters in $\Sigma$):

\begin{remark}
 The language $L_1$ of prefixes of the infinite word $a_1a_2\cdots$ is $2$-testable.
\end{remark}

Let us denote by $W_1$ the word $a_1a_2$ and by $W_i$ for every $i \ge 2$ the word $a_1a_p\cdots a_{p-1}a_p$. We denote by $L_2$ the language of all words $W_1\cdots W_p$ for every $p$ (that is $L_2$ contains the concatenations of the the $p$ first words $W_i$ for all integer $p$).
%Let $L$ be the language containing all strings $(s_{p^c})_p$ where $s_1=a_1a_2$ and for every $p$, $s_p=s_{p-1}a_1a_p\cdots a_{p-1}a_p$. 
Observe that strings in $L_2$ are exactly the ones starting with $a_1a_2$, ending with $a_{p-1}a_p$, and whose subsequences of length $3$ are of the following shape:
\begin{itemize}
    \item $a_{k-1}a_\ell a_k$ for some $k<\ell$
    \item $a_\ell a_k a_\ell$ for some $1<k<\ell$
    \item $a_{\ell-1}a_1a_{\ell}$ for some $1<\ell$
    \item $a_{k-1}a_ka_1$ for some $1<k$
\end{itemize}

\begin{remark}
 The language $L_2$ is $2$-testable.
\end{remark}

%%%%%%%%%%%%%%%%%%%%%%%%%%%%%%%%%%
\subsection{Constructions reaching the bound $\Theta(S_{c,d}(n))$.}\label{subsec:maxdiam_constr}
%%%%%%%%%%%%%%%%%%%%%%%%%%%%%%%%%

The goal of this part is to provide a local checker whose maximum diameter is exactly the one of Theorem~\ref{thm:no_gap}. We distinguish two cases depending the value of $d$.
\smallskip

\noindent
\textbf{Case $d\geqslant 3$.} 
We start with the easier case $d\geqslant 3$. Since $L_1$ is locally testable, for every possible encoding $f$, there is a special local checker in $\mathcal{L}_{c,d}^*$ accepting exactly the trees $T[s]$ for $s\in L$. 

To obtain the bound of Theorem~\ref{thm:no_gap}, it remains to describe $f$. Consider an enumeration of all trees of height $d-1$ by increasing number of vertices, and define $f(a_i)$ as the $i$-th such tree. By construction, the number of vertices of $f(a_i)$ is the smallest $k$ such that $i\leqslant t(d-1,k)$. By Theorem~\ref{thm:equiv}, the size of $f(a_i)$ is $k=\Theta(\log^2 i)$ if $d=3$ and $k=\Theta(g_d(\log i))$ otherwise. 
Now let $T$ be an accepted tree. By Lemma~\ref{lem:d-caterpillar} it is a $d$-caterpillar and by construction of the local checker, the nodes of the backbone are associated to the sequence of letters $a_1\ldots a_p$ of $L_1$ for some integer $p$. 
Now observe that $T[a_i]$ has diameter at most $d$ and contains at most $3 \cdot f(a_p)+2$ vertices (since we attach on each vertex three copies of the same tree plus at most $2$ leaves). Thus $T$ has diameter $\Theta(p +d)$ and its number of vertices is $\Theta(p\cdot |f(a_p)|)$.
%Now observe that $T[a_1\cdots a_p]$ has diameter $\Theta(p)$ and its number of vertices is $\Theta(p\cdot |f(a_p)|)$. 
Plugging in the estimates for $|f(a_p)|$ concludes.

\smallskip
\noindent
\textbf{Case $d=2$.} 
%Let $L$ be the language containing all strings $(s_{p^c})_p$ where $s_1=a_1a_2$ and for every $p$, $s_p=s_{p-1}a_1a_p\cdots a_{p-1}a_p$. Observe that strings in $L$ are exactly the ones starting with $a_1a_2$, ending with $a_{p^c-1}a_{p^c}$ and whose subsequences of length $3$ are of the following shape:
%\begin{itemize}
%    \item $a_{k-1}a_\ell a_k$ for some $k<\ell$
%    \item $a_\ell a_k a_\ell$ for some $1<k<\ell$
%    \item $a_{\ell-1}a_1a_{\ell}$ for some $1<\ell$
%    \item $a_{k-1}a_ka_1$ for some $1<k$
%\end{itemize}
%Therefore $L$ is locally testable, 
Since $L_2$ is locally testable,
and for every choice of $f$, there is a local checker in $\mathcal{L}_{c,2}^*$ accepting exactly the trees $T[s]$ for $s\in L_2$. 

To obtain the bound from Theorem~\ref{thm:no_gap}, it remains to describe the encoding $f$ of each letter and to prove that it gives the required bound. 

Let us start with the colorless case $c=1$. We define $f(a_p)$ as the star on $p+1$ vertices, rooted at its center. Consider the word $s_p=W_1\cdots W_p$. Observe that since $s_p$ has length $p(p-1)$, $T[s_p]$ is a tree of diameter $\Theta(p^2)$. Moreover, each letter $a_i$ appears $(p-1)$ times in $s_p$, hence $T[s_p]$ has $\Theta((p-1)\cdot (1+\cdots+p))=\Theta(p^3)$ vertices. In particular, our local checker has exact diameter $n^{2/3}$, as expected.  

To extend this result to the colored case $c>1$, we just have to adapt the encoding of each letter. We define $f(a_i)$ as a star whose center gets color $1$, and with $x_j$ leaves of color $j$ for $j\in[1,c]$, where $(x_1,\ldots,x_c)$ is the $i$-th (in lexicographic order) $c$-tuple whose entries are non-increasing. Observe that the number of such tuples whose first element is at most $p$ is ${ p+c\choose c}=\Theta(p^c)$, hence each letter in $s_{p^c}$ is encoded with at most $cp$ leaves. In particular, the trees accepted by our local checker have diameter $\Theta(p^{2c})$ and $\Theta(cp\cdot p^{2c})$ nodes, yielding the bound $\Theta(n^{2c/(2c+1)})$ as required.

%\textcolor{red}{À garder si on veut justifier qu'on peut checker localement que c'est $a_{i+1}$ (section 3?)\\
%(note that each node can then easily check if it was given a valid encoding). Given the encoding of $a_i$, one can the easily obtain the encoding of $a_{i+1}$ by adding $1$ to the rightmost position where it is possible without breaking the non-increasing condition, and putting zeroes afterwards. 
%}

%%%%%%%%%%%%%%%%%%%%%%%%%
\subsection{Extension to exact diameter $O(S_{c,d}(n))$}\label{sec:exactdiam}
%%%%%%%%%%%%%%%%%%%%%%%%%

The goal of this part is to prove Theorem~\ref{thm:no_gap} that is to obtain a checker of exact diameter $D(n)$ for every function $D(n)=O(S_{c,d}(n))$. The case $d=1$ being trivial, we will assume in the rest of this section that $d \ge 2$.

Let $D$ be any function such that $D(n)=O(S_{c,d}(n))$. We consider a slightly modified version of the local checkers $L\in\mathcal{L}^*_{c,d}$ constructed in Section~\ref{subsec:maxdiam_constr}. Roughtly speaking, all the nodes have exactly the same behavior but the last node of the backbone, corresponding to the last letter of the word. This node will be allowed to have arbitrarily many leaves (as long as the modulo $3$ condition stays satisfied), and will carry an additional verification. Namely it will check that we added the right amount of leaves to get the correct dependency between diameter and number of nodes. Unfortunately, such a modification does not work directly, and we need to slightly modify our checkers to make it work.

For a node of the backbone corresponding to a letter $a$, we create six copies of $f(a)$ instead of three. This at most doubles the total number of vertices. Now, we link the last node $u$ of the backbone with the center of three stars of the same size. (And we moreover add $0,1$ or $2$ leaves as before depending on the position on the backbone). 

Note that all the nodes must have degree equal to $0,1$ or $2$ modulo $6$ except the last node of the backbone which can have degree $3,4$ or $5$ modulo $6$. And all the nodes of the backbone can check the validity of their degree. Moreover, all the nodes can indeed recover their own letter as well as the letter of their neighbors as in the proof of Lemma~\ref{lem:d-caterpillar}. Indeed, since we add to the last node of the backbone only three copies of the same star, this star will be ignored by the penultimate node of the backbone when it reconstructs the subtree. In the case $d=2$, the argument is a bit more subtle: indeed the neighbor on the backbone does not fully see the subtree attached on the last vertex of the backbone but sees that the number of vertices attached to it is $3$ more than what it should be and does not reject in that case.

To conclude, let us explain when the last vertex of the backbone accepts (the acceptance rule is not modified for the other vertices). The goal is to artificially increase the number of vertices of the trees $T[w]$ for $w\in L_2$, without changing their diameter, in order to reach the target diameter $D$. Note that, in both languages $L_1$ and $L_2$, the last node of the path is able to determine the length of the path by simply knowing the last two letters. So the last node of the backbone computes the diameter $\delta$ of $T[w]$ and the number $k$ of vertices in the tree before it. 
%Note that it can recover the last letter $a_p$ of the word, hence derive $p$, and get estimates for $\delta$ (namely $p$ when $d>2$ and $p^2$ when $d=2$) and for $k$. 
It then accepts when its number $r$ of leaves satisfies $D(r+k)=\delta$.

%%%%%%%%%%%%%%%%%%%%%%%%%%%%%%%%%%%%%

%%%%%%%%%%%%%%%%%%%%%%%%%%%%%%%%%%%%%%%%%%%%%%%%%%%%%%%%
%%%%%%%%%%%%%%%%%%%%%%%%%%%%%%%%%%%%%%%%%%%%%%%%%%%%%%%%

\section{Minimum diameter at distance $1$ -- Proof of Theorem~\ref{thm:radius-1-MIN}}
\label{sec:min-diameter}

%%%%%%%%%%%%%%%%%%%%%%%%%%%%%%%
\ThmRadiusOneMin*
%%%%%%%%%%%%%%%%%%%%%%%%%%%%%%%

A first idea consists in just applying our pumping argument similarly to the proof of Lemma~\ref{lem:pumping_lemma}, but in reverse, in order to shorten long enough paths until they get constant size. However, this may actually decrease a lot the total number of vertices at the same time, impeding us to construct an infinite family of trees with the claimed diameter. To solve this issue, we have to get more control on which subtrees we remove while de-pumping. More precisely, we will only de-pump subtrees that are almost paths. To ensure that long enough paths exist, we first prove that local checkers accepting trees of arbitrarily large degree have constant minimum diameter. We then consider local checkers $L$ accepting trees of bounded degree (which already have super-logarithmic diameter), and investigate the edges with a given view to exhibit some structure in the trees accepted by $L$. We then use this structure to show that $L$ must accept trees with a very specific shape, namely that look like complete binary trees or rakes (defined below). 

%%%%%%%%%%%%%%%%%%%%%%%%%%%%
\subsection{Constructions}
%%%%%%%%%%%%%%%%%%%%%%%%%%%%

Before diving into the proof, we will prove that Theorem~\ref{thm:radius-1-MIN} is essentially tight, meaning that we can obtain local checkers that accept trees of constant diameter, linear diameter or diameter $\Theta(n^{1/k})$ for every $2 \le k \le c/3$.
\medskip

\noindent\textbf{Linear diameter.} The local checker in $\mathcal{L}_{1,1}$ where every node accepts whenever its degree is at most $2$ accepts exactly the class of paths, hence has linear exact (and then minimum) diameter.
\medskip

\noindent\textbf{Constant diameter.}
 The local checker in $\mathcal{L}_{1,1}$ that where every node accepts regardless of its neighborhood accepts trees of diameter at most $2$ for every possible size of trees.
% Assume that we have two colors black and white.
%The local checker in $\mathcal{L}_{2,1}$ where every black node accepts whenever it has degree 1 and every white node has only black neighbors accepts exactly the class of stars, hence has constant exact diameter.
\medskip

\noindent\textbf{Logarithmic diameter.}
The local checker in $\mathcal{L}_{1,1}$ where a node accepts if and only if its  degree is $1$ or $3$ accepts exactly binary trees. Hence its minimum diameter is logarithmic (even if its maximum diameter is linear) since the minimum diameter of such a tree is reached when the tree is complete.
\medskip

\begin{figure}[!ht]
    \centering
    \includegraphics[scale=0.7]{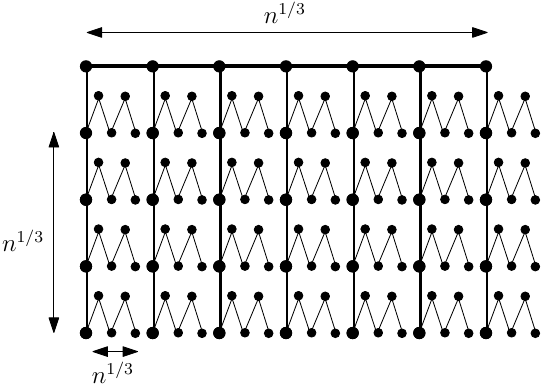}
    \caption{A $3$-rake on $n$ vertices of diameter $\Theta(n^{1/3})$. The deletion of the top path leaves $n^{1/3}$ $2$-rakes.}
    \label{fig:rake}
\end{figure}

\noindent\textbf{Polynomial diameter.}
For each integer $k$, let us introduce the class of $k$-rakes. The \emph{$1$-rakes} are paths, and \emph{$k$-rakes} are trees $T$ containing a path $P$ whose removal yields a forest whose connected components are $(k-1)$-rakes and such that each vertex of $P$ is attached to at most one connected component of $T \setminus P$ (see Figure~\ref{fig:rake} for an illustration). We prove the following in the appendix.
\begin{lemma}
\label{lem:krakes}
    The class of $k$-rakes has minimum diameter $\Theta(n^{1/k})$ and is accepted by a local checker in $\mathcal{L}_{3k,1}$. 
\end{lemma}

%%%%%%%%%%%%%%%%%%%%%%%%%%%%%%%%%%%
\subsection{Proof of Theorem~\ref{thm:radius-1-MIN}}
%%%%%%%%%%%%%%%%%%%%%%%%%%%%%%%%%%%

We now proceed with the proof of Theorem~\ref{thm:radius-1-MIN}, and first take care of trees of arbitrary large degree.

\begin{lemma}
    Let $c$ be an integer and $L\in\mathcal{L}_{c,1}$ accepting trees of arbitrary large degree. Then the minimum diameter of $L$ is constant.
\end{lemma}

\begin{proof}
    Given two colors $a,b$, define $T_{a,b}$ as a minimum-sized tree accepted by $L$ containing an edge $v_av_b$ such that $v_a$ has color $a$ and $v_b$ has color $b$ (if such a tree exists). Denote by $\alpha$ the largest size of $T_{a,b}$ when $a,b$ run across all possible colors. 

    By hypothesis, $L$ accepts an infinite family of $c$-colored trees $(T_n)_n$ such that each $T_n$ contains a vertex $u_n$ of degree at least $n$. Root each $T_n$ at $u_n$. Now, for each edge $u_n u$ of $T_n$, let $a$ be the color of $u_n$ and $b$ the color of $u$, and successively graft $T_{a,b}$ in $T_n$ at $(uu_n,v_av_b)$. By Lemma~\ref{lem:grafting}, the resulting tree $T'_n$ is accepted by $L$. Moreover, every tree $T'_n$ has at least $n$ nodes and diameter at most $2\alpha$, which concludes.    
\end{proof}

We may now assume that each local checker accepts only graphs of bounded degree $\Delta$. In particular, each accepted tree $T$ satisfies $|T|\leqslant \Delta^{\diam(T)}$, hence the diameter is at least logarithmic (and at most linear) with respect to the size\footnote{Note that even if we restrict to the case of bounded degree, we cannot directly use the LCL machinery since the context is not exactly the same. We are looking at possible diameters of accepted trees and not looking at the checking of some property.}. 
To prove the remaining parts of Theorem~\ref{thm:radius-1-MIN}, we study the structure of the edges with the same view, and get a criterion proving that $L$ accepts trees that look like complete binary trees or $k$-rakes for some $k$, and thus must have minimum diameter $O(\log n)$ or $O(n^{1/k})$. 

Given a local checker $L$, we say that a pair of colors $(c_1,c_2)$ is \emph{useful} if $L$ accepts some tree containing a path $u_1,\ldots,u_p$ ($p\geqslant 4$) where $u_1,u_{p-1}$ have color $c_1$ and $u_2,u_p$ have color $c_2$. 

We define the binary relation $<_L$ on the useful pairs of colors as follows: $(c_1,c_2)<_L (d_1,d_2)$ if $L$ accepts a $c$-colored tree $T$ that can be rooted such that, when orienting the edges towards the leaves, there are three arcs  $u_1u_2, v_1v_2$ and $w_1w_2$ such that $u_2$ is an ancestor of $v_1$ and $w_1$, $v_2$ is not an ancestor of $w_2$, both $u_i,v_i$ have color $c_i$ for $i=1,2$, and $w_i$ has color $d_i$ for $i=1,2$. When this condition is satisfied, we moreover say that the rooted tree $T$ \emph{witnesses} that $(c_1,c_2)<_L (d_1,d_2)$. (In other words, it means that $u_1,u_2$ are ancestors of the four other vertices while $v_1v_2$ and $w_1w_2$ lie in different subtrees and are incomparable in $T$) (see Figure~\ref{fig:partialorder} for an illustration).

\begin{figure}[!ht]
    \centering
    \includegraphics[scale=0.6]{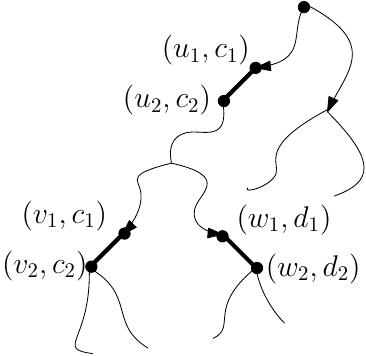}
    \caption{A tree $T$ witnessing  $(c_1,c_2) <_L (d_1,d_2)$. Vertices are labeled with their name and color.}
    \label{fig:partialorder}
\end{figure}

One can easily check that $<_L$ is transitive (a proof of this statement is given in Appendix~\ref{app:transitive}). To get the remaining cases of Theorem~\ref{thm:radius-1-MIN}, we prove in Appendix~\ref{app:log} that if $<_L$ is not a strict partial order, then $L$ has minimum diameter $\Theta(\log n)$. Otherwise, we claim that $L$ has minimum diameter $\Theta(n^{1/k})$ where $k$ is the length of the longest chain for $<_L$ (in particular $k\leqslant c^2$). This is a consequence of the two following results, whose proofs are deferred to the appendix. 

\begin{lemma}
    \label{lem:chaintomindiam}
    If $<_L$ has a chain of size $k$, then the minimum diameter of $L$ is $O(n^{1/k})$.
\end{lemma}

\begin{lemma}
    \label{lem:mindiamtochain}
    If $L$ accepts an $n$-vertex tree of diameter at most $n^{1/k}/\Delta^{(1+1/k)c^2}$, then $<_L$ has a chain of size $k+1$.
\end{lemma}

%%%%%%%%%%%%%%%%%%%%%%%%%%%%%%%%%%%%%%%%%%%%%%%
\bibliographystyle{plain}
\bibliography{biblio}
%%%%%%%%%%%%%%%%%%%%%%Story%%%%%%%%%%%%%%%%%%%%%%%%%

\appendix 

%\input{old-stuff}*

%%%%%%%%%%%%%%%%%%%%%%%%%%%%%%%%%%%%%%%%%%%%
\section{Further related work}
\label{app:further-related-work}
%%%%%%%%%%%%%%%%%%%%%%%%%%%%%%%%%%%%%%%%%%%%

\paragraph*{Network science perspective}

Network science has a long history of linking local and global properties of networks. A typical example are scale-free networks where the power-law degree distribution is related to clustering and small world phenomenon, and for which generative models have been introduced (such as preferential attachment).
We refer to the book~\cite{Lewis11} for an introduction to the topic. In general, this perspective differs from ours in two ways: the properties considered are global (\emph{e.g.} the degree distribution) thus cannot be checked locally, and the generative models are dynamic processes, while we study static processes (in other words, we are interested in maintaining a structure, not in creating it). 
Papers outside of network science also consider processes to ``grow a graph'', see in particular~\cite{MertziosMSST22}.

\paragraph*{Graph theory perspective}

Graph theory also studies relations between local conditions and global behavior, in a more combinatorial way. For example, there is a wealth of works showing that forbidding small structures in a graph implies the existence of nice decompositions, of bounds on the coloring number etc. Another related direction are theorems à la Dirac: if all nodes have degree at least some function of the number of nodes, then the graph is Hamiltonian or connected. Again, this is quite different from our direction, in particular the diameter is not a topic of interest in this area.

\paragraph*{Distributed computing perspective}

We have already mentioned several works on LCLs, and how they are related to our paper.
For completeness, let us also mention that the complexity landscape has also been studied in rooted trees~\cite{BalliuBCOSST23} and in trees for node-averaged complexity~\cite{BalliuBKOS23}. 
The topic of checking the configuration of a network is central in self-stabilization~\cite{AltisenDDP19}, where it is essential to be able to detect inconsistencies in the computed data-structure, but not in the network topology, which differs from our work. A related field, that does study the network itself is local certification, where one considers labels helping the nodes to check given properties. (We refer to \cite{Feuilloley21} for an introduction to the topic, and to \cite{BousquetFZ24} for a recent paper surveying the works related to graph structure.)

\section{Proofs of Section~\ref{sec:maxdiam}}

\subsection{Proof of Lemma~\ref{lem:grafting}}

    Using the notations of the definition, suppose that $T$ and $T'$ are accepted by $L$, and consider the graft $T''$ of $T'$ in $T$ at $(uv,u'v')$. We claim that each node $w$ of $T''$ has the same view as its copy in $T$ or $T'$. Assume by symmetry that $w$ was in $T$, and root $T''$ at $w$. Then its neighborhood at distance $d$ in $T''$ consists in a copy of its neighborhood at distance $d$ in $T$, except if it contains $v'$. In that case, the subtree rooted at $v'$ is $T'_{v'}$. Recall that $uv$ (in $T$) and $u'v'$ (in $T'$) have the same view at distance $d$, hence the view of $v'$ in $T'_{v'}$ and of $v$ in $T_v$ at distance $d-1$ are the same. Thus the the view of $w$ in the subtree rooted at $v'$ (in $T''$) is the same as the view of subtree rooted at $v$ (in $T$), hence it accepts.
    Therefore, $T''$ is accepted by $L$.

\subsection{Proof of Lemma~\ref{lem:pumping_lemma}}

    Let us denote by $C_1, C_2$ and $C_3$ the connected components of $T \setminus \{ uv,xy \}$ containing respectively $u$ for $C_1$, $v$ and $x$ for $C_2$ and $y$ for $C_3$.

    Let us now define an infinite collection of trees $(T_i)_{i \in \mathbb{N}}$ such that, for every $i$, the diameter of $T_{i+1}$ is larger than the  one of $T_i$ and the size of $T_{i+1}$ is the size of $T_i$ plus the size of $C_2$.

    We set $T_1=T$. We define $T_2$ as the tree obtained by grafting $T$ in $T$ at $(xy,uv)$. Note that, abusing notations, the resulting tree $T_2$ contains the three following edges: $uv$ (in one of the trees $T$), $xv$ (corresponding to $x$ in the first $T$ and $v$ in the second) and $xy$ (in the second tree). We denote by $x_2y_2$ the edge $xy$. Now let $i \ge 3$.  and we define $T_{i+1}$ as the tree obtained by grafting $T$ in $T_i$ at $(x_2y_2,uv)$. 
    Observe that this operation creates a new copy of $xy$ in $T_{i+1}$ that we denote by $x_{i+1}y_{y+1}$ (lying in the part of $T$ that is added to $T_i$). By Lemma~\ref{lem:grafting}, each $T_i$ is accepted by $L$. 
    
    The conclusion follows from the fact that $T_i$ has at most $|T|+(i-1)\cdot |C_2|$ nodes and diameter at least $i\cdot d_T(x,v)$.

\section{Proof of Lemma~\ref{lem:d-caterpillar}}

Consider a local checker that can see at distance $d+1$. In particular, a vertex $u$ can see all the neighbors of all its neighbors at distance at most $d$. The vertex $u$ can then iteratively remove all the leaves $d$ times. After all these deletions, if the vertex is not eliminated, it should have at most $2$ neighbors (otherwise it rejects) which are its (at most) two neighbors on the backbone. This proves the first and second points. 

Let us denote by $T_u$ the view of $u$ rooted at $u$ where the subtrees rooted on the neighbors of $u$ in the backbone are removed. The vertex $u$ first checks its degree in $T_u$ and checks that it is coherent with the degree of its neighbors on the backbone (if it has degree $1$ mod $3$, its two neighbors must have degree $0$ and $2$).

After removing leaves attached to $u$ in $T_u$ in order to reach a degree $0$ mod $3$, the vertex $u$ can check that its pending subtrees can be partitioned into three copies of the same forest in order to determine its letter. Since it sees at distance $d+1$ it can run a similar process to determine the letter of its neighbor. Note that this is not completely immediate since, even if it sees the trees of depth $d$ rooted on its neighbors it cannot check that leaves are indeed leaves. However, after the removal of 0, 1 or 2 attached leaves on its neighbor, the attached trees must consist of three times the same forest plus an additional tree that corresponds to the beginning of the subtree attached on the neighbor at distance $2$ on the backbone. In particular the degree is equal to $1$ modulo $3$. So after grouping the trees three by three, it only remains one tree which corresponding to the rest of the backbone. Removing it allows $u$ to know the tree associated with its neighbor, and thus its letter, which completes the proof of the fourth point.

For the endpoints of the backbone, there is only one long path attached to them, which can indeed be easily determined. This proves the third point.
%\begin{itemize}
%    \item Every node $u$ counts how many of the subtrees rooted at its neighbors have depth at most $d-1$ (this is done by checking that all nodes at distance $d$ have degree $1$).
%    \item If this number is at most 1, it accepts, and it rejects if it is at least 3.
%    \item If this number is $2$, say for two neighbors $v,w$, it checks first whether one of the subtrees rooted at $v$ or $w$ (say $v$) is a path of length $d$. In that case, $u$ is an endpoint of the backbone and it can check that $u$ and $w$ have distinct degrees modulo $3$, and recover the character in its view. Moreover, it can also recover the character $a$ associated with $w$ since $f(a)$ is the subtree of $w$ of height at most $d$ that is repeated three times.
%    \item Otherwise, $u$ checks that $u,v,w$ have pairwise distinct degrees modulo 3, and similarly recovers the characters encoded in the subtrees of $u,v,w$ (if any such process fails, the node rejects).
%\end{itemize}

\section{Consequences on generalized LCL -- Proof of Corollary~\ref{coro:LCL}}
\label{app:coroLCL}
%%%%%%%%%%%%%%%%%%%%%%%%%%%%%%%%%%%%%

%%%%%%%%%%%%%%%%%%%%%%%
\CoroLCL*
%%%%%%%%%%%%%%%%%%%%%%%

    Let us first discuss slight modifications of the checkers we have manipulated to get our upper bounds in Section~\ref{sec:exactdiam}.
    We will prove that we can modify checkers to ensure that if our original special checker $L$ accepts a tree $T[w]$, then the modified checker accepts all the trees $T[w']$ where $w'$ is obtained from $w$ by repeating a constant number of times some letters.

    Let $k \ge 1$. A word $W'$ is a $k$-subdivision of $W=a_1\cdots a_n$ if it is a subword of $a_1^k\cdots a_n^k$ and it contains $W$ as a subword. The $k$-subdivision of a language $S$ is the language of all the $k$-subdivisions of words of~$S$. (Note that $S$ is the $1$-subdivision of $S$). Corollary~\ref{coro:LCL} relies on the fact that $k$-subdivision preserves recognition by special local checkers, up to increasing the checkability radius.
    \begin{claim}\label{clm:generalizedLCL}
         Let $r\ge k \ge 2$ and $L\in \mathcal{L}^*_{c,d}$ accepting a $r$-testable language $S$ such that no word of $S$ has identical consecutive letters. Then there is a special local checker $L^*$ in  $\mathcal{L}^*_{c,d+rk}$ accepting the $k$-subdivision of $S$.
    \end{claim}

\begin{proof} 
We consider the same checker $L$ with slight modifications. Each node $u$ now can recover all letters at distance $kr-1$, and remove all the repetitions it sees. Since no word of $S$ has identical consecutive letters, and each can be repeated at most $k$ times in the subdivision, $u$ now has access to the letters at distance $r-1$ from it in the un-subdivided word and run the verification process from $L$ (since $S$ is $r$-testable). 
\end{proof}
%    We simply have to This modification can be done by at most doubling the checkability radius of the checker, and making it simulate the original checker on the graph where the subdivisions are undone.
    
    We may now proceed with the proof of Corollary~\ref{coro:LCL}. 
    \begin{proof}[Proof of Corollary~\ref{coro:LCL}]
    Let us now consider the languages $L_1$ and $L_2$ of the previous sections which are $2$-testable. Note that there is no repetition of letters in neither $L_1$ nor $L_2$. By Claim~\ref{clm:generalizedLCL}, there is a special checker that can recognize their $2$-subdivisions. 
    %Note that the modification  does not change the asymptotic diameter as a function of $n$.
    
    Now, we define a generalized-LCL, which is a generalization of $2$-coloring, in the following way. Each node must choose an output that is either 0, 1, or empty.  
    The verification algorithm checking the correctness of the output is the following. On a node $v$,
    \begin{itemize}
        \item The output of all the vertices that are not in the backbone is empty while all the vertices in the backbone should output $0$ or $1$. 
        \item If $v$ outputs $0$ or $1$, it cannot have a neighbor with the same output (i.e. the non-empty output labels define a 2-coloring.)
    \end{itemize}
    
    We claim that in the trees accepted by the modified local checker, this is a global problem. 
    Consider two nodes at the endpoints of the maximum path, and suppose they have a view asymptotically smaller than the diameter. 
    Now, these nodes cannot recover the middle of the word encoded by the tree. 
    In this middle part of the word, either delete a letter that was repeated or delete a repetition. This leads to a new tree, where the endpoints of the backbone have still the same view, hence should output the same. The $2$-colorings of one of the two trees  must be incorrect, which is a contradiction.
\end{proof}

Finally note that in the proof above, we have decided to increase a bit the checkability radius of the local checker. 
If we allow oursleves to have two input colors, then we can keep the same radius. 
In the instances we consider, all nodes will have input 0, except possibly the first node of the backbone which can have either 0 or 1.
Then we enforce that first node should have the same input and output label. Since the nodes on the other side of the tree do not know about this input, 2-coloring is again a global problem. 

\section{Proofs of Section~\ref{sec:min-diameter}}
\subsection{Proof of Lemma~\ref{lem:krakes}}

Constructing inductively a $k$-rake by choosing each time a path $P$ on $\ell$ vertices, we obtain a $k$-rake with $\ell^k$ vertices and of diameter $k\ell$. Hence the class of $k$-rakes has minimum diameter $O(n^{1/k})$. This bound is actually tight: this is clear for $1$-rakes. For $k \ge 2$, let $R$ be a $k$-rake and $P$ be a path whose deletion leaves a collection of $(k-1)$-rakes. If $P$ has length at least $n^{1/k}$, the conclusion follows. So we can assume that $|P|<n^{1/k}$. Since $R \setminus P$ contains at most $|P|$ connected components, one of the $(k-1)$-rakes has size at least $n^{(k-1)/k}$. And by induction hypothesis, such a $(k-1)$-rake has diameter at least $n^{1/k}$.

Let us now prove that for every $k>0$, there exists a local checker $L\in\mathcal{L}_{3k,1}$ accepting exactly the class of $k$-rakes as long as $c \ge 3k$. The $3k$ colors will be represented as pairs of integers $(i,j)$ with $1\leqslant i \leqslant k$ and $j\in\{1,2,3\}$. Each node checks that it has degree at most three. Moreover each node of color $(i,j)$ checks that it has a neighbor $(i+1,1)$ (if $i<k$) and that its other neighbors have colors either $(i,j-1\mod 3)$ and $(i,j+1\mod 3)$, or $(i,j+1\mod 3)$ and $(i-1,*)$ or only $(i,j-1\mod 3)$.

\subsection{$<_L$ is transitive}
\label{app:transitive}
    Assume that $(c_1,c_2) <_L (d_1,d_2)$ and $(d_1,d_2) <_L (e_1,e_2)$. Denote by $T,T'$ the (rooted) trees that witness the relation $(c_1,c_2) <_L (d_1,d_2)$ (resp. $(d_1,d_2) <_L (e_1,e_2)$). In particular, $T$ contains two arcs $u_1u_2,v_1v_2$ of colors $(c_1,c_2)$ and an arc $w_1w_2$ of colors $(d_1,d_2)$, and $T'$ contains the arcs $x_1x_2,y_1y_2$ of colors $(d_1,d_2)$ and $z_1z_2$ of color $(e_1,e_2)$ with the ancestor relation of the definition of $<_L$.

    Construct the tree $T''$ by grafting $T'$ in $T$ at $(x_1x_2,w_1w_2)$. By Lemma~\ref{lem:grafting}, $T''$ is accepted by $L$. Moreover, by construction of $T''$, the copy of $u_2$ is an ancestor of $z_1$ but not $v_2$, hence $(c_1,c_2)<_L (e_1,e_2)$. 

\subsection{The logarithmic case}
\label{app:log}
    If $<_L$ is not a strict partial order, since it is transitive, there must be a pair of colors $(c_1,c_2)$ such that $(c_1,c_2)<_L (c_1,c_2)$. In particular there is a rooted tree $T$ with three arcs $u_1u_2, v_1v_2, w_1w_2$, all of colors $(c_1,c_2)$ such that $u_2$ is an ancestor of $v_1$ and $w_1$ but $v_2,w_2$ are not ancestors of each other. 

\begin{figure}
    \centering
    \includegraphics[scale=.65]{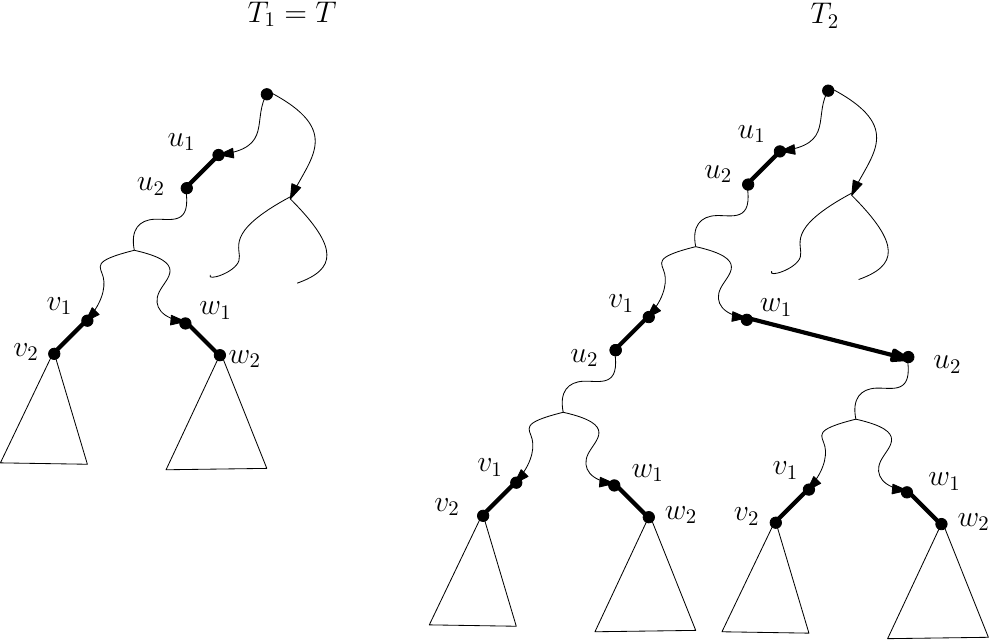}
    \caption{Illustration of the proof of Appendix~\ref{app:log}. The trees $T_1$ and $T_2$ are depicted.}
    \label{fig:logcase}
\end{figure}

    From $T$, we construct a sequence of trees $(T_i)_i$ that look like binary complete trees. Each $T_i$ will contain $2^i$ copies of the subtrees of $T$ rooted at $u_2$. Let $T_1=T$ and define $T_{i+1}$ by successively grafting $T$ in $T_i$ at $(xy,u_1u_2)$, where $xy$ runs across all $2^i$ copies of $v_1v_2$ or $w_1w_2$ in $T_i$. Observe that $L$ accepts each $T_i$ by Lemma~\ref{lem:grafting} (see Figure~\ref{fig:logcase} for an illustration). Moreover, the diameter of $T_i$ increases by a constant (at least one and at most the diameter of $T$) compared to $T_{i-1}$ and the number of nodes increases by at least $2^i$ and at most $|T| 2^i$. Thus $T_i$ has diameter $\Theta(i)$ and $\Theta(2^i)$ nodes, which concludes.

\subsection{Proof of Lemma~\ref{lem:chaintomindiam}}

We basically show that if $<_L$ has a chain of size $k$, then $L$ accepts trees with a $k$-rake-like structure, hence its minimum diameter is $O(n^{1/k})$.

Consider a chain of $k$ pairs of colors $(c_1,c'_1) <_L \cdots <_L (c_k,c'_k)$. Denote by $T^{(i)}$ a rooted tree and $u_1^{(i)}u_2^{(i)},v_1^{(i)}v_2^{(i)},w_1^{(i)}w_2^{(i)}$ its vertices witnessing that $(c_i,c'_i)<_L (c_{i+1},c'_{i+1})$. Finally, let $T^{(k)}$ be a tree witnessing that $(c_k,c'_k)$ is useful, that is $T^{(k)}$ contains a path between two edges $u_1^{(k)}u_2^{(k)}$ and $v_1^{(k)}v_2^{(k)}$ of colors $(c_k,c'_k)$. 

\begin{figure}
    \centering
    \includegraphics[scale=0.6]{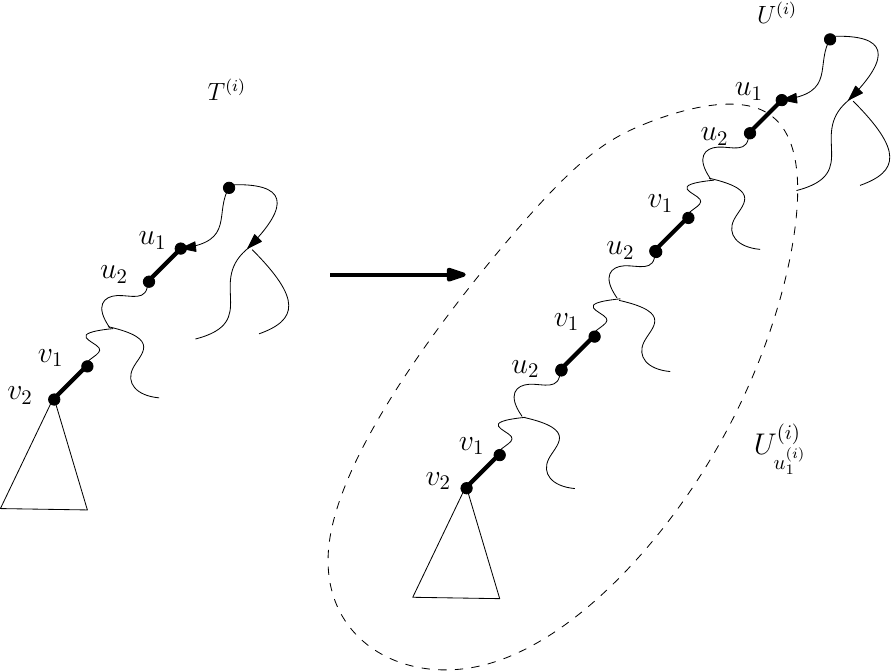}
    \caption{The construction of $U^{(i)}$ from $T^{(i)}$ (some exponents have been removed for readability).}
    \label{fig:Uk}
\end{figure}

Let $N>0$. For every $i$, we start by grafting $T^{(i)}$ in $T^{(i)}$ at $(v_1^{(i)}v_2^{(i)},u_1^{(i)}u_2^{(i)})$ $N$ times, as in the proof of Lemma~\ref{lem:pumping_lemma} (see Figure~\ref{fig:Uk}). This yields a tree $U^{(i)}$ containing a copy of the edge $u_1^{(i)}u_2^{(i)}$ of $T^{(i)}$ and $N$ copies of $w_1^{(i)}w_2^{(i)}$, that is still accepted by $L$ by Lemma~\ref{lem:grafting}.

\begin{figure}[!ht]
    \centering
    \includegraphics[scale=0.65]{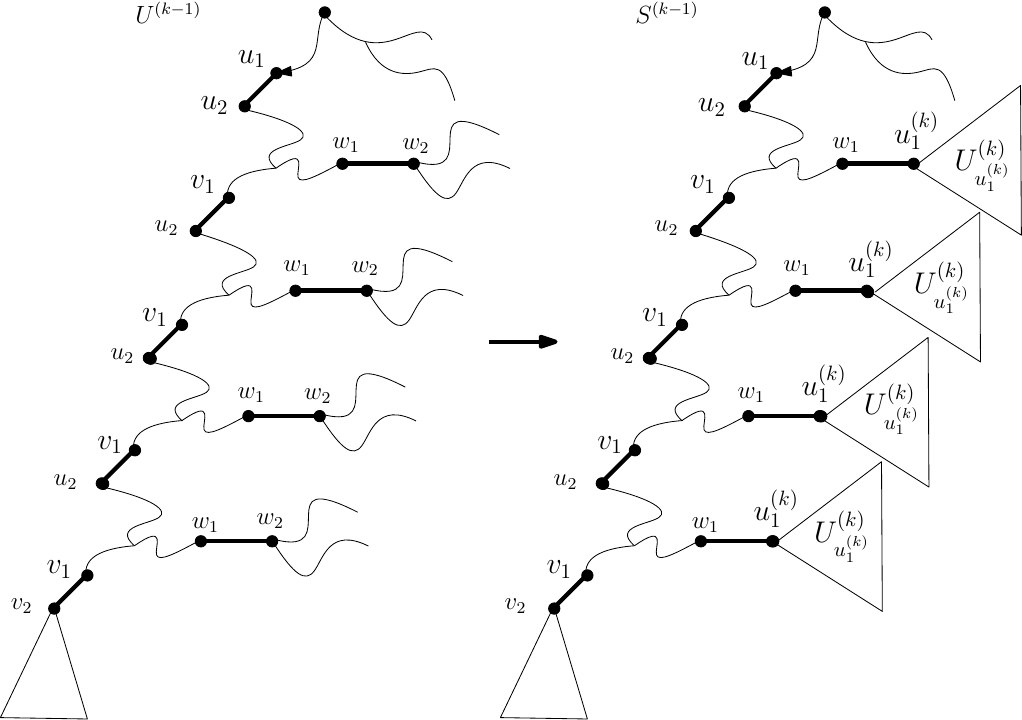}
    \caption{The graph $S^{(k-1)}$ built from $S^{(k)}$. The graph $U^{(k)}_{u_1^{(k)}}$ is the one of Figure~\ref{fig:Uk}.}
    \label{fig:orderproof2}
\end{figure}

Set $S^{(k)}:=U^{(k)}$ and for $i=k-1$ down to $1$ define $S^{(i)}$ as the tree obtained by successively grafting $S^{(i+1)}$ in $U^{(i)}$ at $(u_1^{(i+1)}u_2^{(i+1)}, xy)$ where $xy$ ranges across the $N$ copies of $w_1^{(i)}w_2^{(i)}$ in $U^{(i)}$ (see Figure~\ref{fig:orderproof2}. By Lemma~\ref{lem:grafting}, all these trees are accepted by $L$.

%\textcolor{purple}{We will now graft the graft $S^{i+1}$ in $U_i$ in all these $N$ edges with the view $(c_{i+1},c'_{i+1})$. More formally, for all the copies of the edge $w_1^{(i)}w_2^{(i)}$, we graft $S_{i+1}$ in $U_i$ at $(a_{i+1}b_{i+1},w_1^{(i)}w_2^{(i)})$ (and still denote the resulting tree $U_i$ by abusing notation). We finally denote by $S_i$ the resulting tree and denote by $a_ib_i$ the special edge $ab$ of $U_i$ that still exists in $S_i$. }

For large enough $N$, each $U^{(i)}$ has diameter and size $\Theta(N)$. Therefore, when constructing $S^{(i)}$ from $S^{(i+1)}$, the number of nodes is multiplied by $\Theta(N)$ while the diameter increases by an additional $\Theta(N)$. In particular, $S^{(i)}$ has diameter $\Theta((k-i+1)\cdot N)$ and contains $\Theta(N^{k-i+1})$ vertices. Taking $i=1$ proves that $L$ has minimum diameter at most $O(n^{1/k})$.

%\textcolor{purple}{By definition of $<_L$ this iterated grafting has created a lot of copies of pairs of colors $(c_k,c'_k)$ (one per grafting). Let us denote by $U^{k-1}_t$ the resulting tree. We will now graft $U^k$ in $U_t^{k-1}$ in all these edges. This essentially provides us a $2$-rake (see Figure~\ref{fig:order_rake} for an illustration). If we repeat this operation from $k$ to $1$ we end up on a tree whose shape is essentially a $k$-rake where all the components have diameter $O(N)$.

%Let us now formalize this argument. For $i$ ranging from $k$ to $1$, we operate the following modifications to $T^{(i)}$ (in order):
%\begin{itemize}
%    \item If $i<k$, graft $T^{(i+1)}$ in $T^{(i)}$ at $(w_1^{(i)}w_2^{(i)}, u_1^{(i+1)}u_2^{(i+1)})$.% the subtree of $T^{(i+1)}$ rooted at $u_1^{(i+1)}$.
%    \item Graft $T^{(i)}$ in $T^{(i)}$ at $(v_1^{(i)}v_2^{(i)},u_1^{(i)}u_2^{(i)})$.% the subtree of $T^{(i)}$ rooted at $u_1^{(i)}$.
%    \item The previous item creates a new copy of $v_1^{(i)}$, so we can repeat this grafting. We iterate this until $T^{(i)}$ contains $N$ copies of $T^{(i+1)}$.
%\end{itemize}

%By Lemma~\ref{lem:grafting}, $L$ accepts all the $T^{(i)}$'s. Moreover observe that for large enough $N$, $T^{(i)}$ has now $\Theta(N^{k+1-i})$ vertices, and its diameter is $\Theta(N)$. Therefore, we can conclude by considering $T^{(1)}$ for large enough values of $N$.

\subsection{Proof of Lemma~\ref{lem:mindiamtochain}}

We prove by induction a slightly stronger and technical statement: if $L$ accepts an arbitrarily tree with an $n$-vertex (possibly non-rooted) subtree $T$ of diameter less than $n^{1/k}/\Delta^{(1+1/k)c^2}$, then $T$ contains an edge whose pair of colors is smaller than $k$ other useful pairs (for $<_L$). 

Assume that there is a tree accepted by $L$ with an $n$-vertex subtree $T$ of diameter less than $n^{1/k}/\Delta^{(1+1/k)c^2}$. Root $T$ arbitrarily. We construct an auxiliary tree $T_f$, starting with $T_f:=T$ and successively modifying it as follows. For each pair of colors $(c_1,c_2)$ (by increasing order of $<_L$), while there exists a branch going through two arcs $u_1u_2$ and $v_1v_2$ with $u_i,v_i$ colored with $c_i$, we assume that $u_1u_2$ and $v_1v_2$ are respectively the closest and furthest to the root among all arcs of the branch with these colors, and we graft at $u_1$ in $T$ the subtree of $T$ rooted at $v_1$. Observe that each time this operation is applied, some vertices, namely the descendants of $u_1$ but not of $v_1$, are deleted from $T$. Denote by $T_f$ the tree obtained after this process is finished, and observe that $T_f$ has height at most $c^2$ since no two pairs of colors can repeat on edges of the same branch. In particular, $T_f$ has at most $\Delta^{c^2}$ nodes.

Since we considered the pairs of colors by increasing order of $<_L$, due to the choice of $u_1u_2$ as closest to the root, all the vertices chosen as $u_1$ and $u_2$ at some point during this process are chosen only once, and moreover they cannot be deleted afterwards. In other words, each deleted vertex of $T$ can be associated with an arc $u_1u_2$ of $T_f$. By the pigeonhole principle, there is one arc $u_1u_2$ associated with at least $n/\Delta^{c^2}$ deleted nodes. 

Consider the branch $b$ and its arc $v_1v_2$ chosen when $u_1u_2$ was selected. Observe that the distance between $u_2$ and $v_1$ is less than $n^{1/k}/\Delta^{(1+1/k)c^2}$, hence one of the trees $T'$ pending from $b$ whose root is between $u_2$ and $v_1$ has size more than $n'= n^{1-1/k} \cdot \Delta^{c^2/k}$. 

First assume that $k>1$. Note that $T'$ has diameter less than $n^{1/k}/\Delta^{(1+1/k)c^2} \leqslant {n'}^{1/(k-1)}/\Delta^{(1+1/(k-1))c^2}$. By induction, $T'$ contains an arc $w_1w_2$ whose pair of colors $(d_1,d_2)$ is smaller than $k-1$ other pairs. Denoting by $(c_1,c_2)$ the colors of $u_1u_2$, we get that $(c_1,c_2) <_L (d_1,d_2)$, hence $(c_1,c_2)$ is smaller than $k$ other useful pairs, which concludes the induction.

We may thus assume that $k=1$. Since $n'\geqslant \Delta^{c^2}$, $T'$ must have depth at least $c^2$. On a longest branch, some pair of colors $(d_1,d_2)$ must be repeated. Hence $(d_1,d_2)$ is useful and $(c_1,c_2)<_L(d_1,d_2)$, which concludes.

%%%%%%%%%%%%%%%%%%%%%%%%%%%%%%%%%%%%%%%
\section{Restricting the expressivity: degree-myopic local checkers}
%%%%%%%%%%%%%%%%%%%%%%%%%%%%%%%%%%%%%%%
\label{sec:myopic}
This section is devoted to the proof of Theorem~\ref{thm:classification-myopic}. 

%%%%%%%%%%%%%%%%%%%%%%%%%%%%%
\ThmClassificationMyopic*
%%%%%%%%%%%%%%%%%%%%%%%%%%%%%

The proof follows a two-step approach: first, we restrict the possible values for the parameters $a_i,b_i$'s, by discarding values that either yield degree-myopic checkers with constant or linear maximum diameter or can be modified without changing the maximum diameter. For example, each tree contains leaves, hence we must have $b_{leaves}\geq 1$. Then, we classify the possible diameter depending on the remaining free parameters.

%%%%%%%%%%%%%%%%%%%%%%%%%%%%%%%%%%
\subsection{Structural properties}
\label{subsec:structural-properties}
%%%%%%%%%%%%%%%%%%%%%%%%%%%%%%%%%%

First, we note that similarly to what happened earlier in the paper, the regimes $O(1)$ and $\Omega(n)$ are uninteresting. It is easy to find degree-myopic local checkers in these regimes: a checker where non-leaf nodes can only see leaves leads to stars; and a checker where every non-leaf node accepts only two neighbors of equal-degree leads to paths.
Therefore, we focus on the set of degree-myopic local checkers whose maximum diameter lies between $\omega(1)$ and $o(n)$, that we denote by $\midd{}$. And prove a series of claims on that set.

\begin{claim}
    \label{clm:aD+1=0}
    For any local checker in $\midd{}$, we must have $a_{Degree+1}=0$. 
\end{claim}

\begin{proof}
    Consider a tree accepted by a local checker with $a_{Degree+1}>0$. Since it is finite, then there is a node $v$ of maximum degree $\Delta$.
    If $\Delta=1$, then the tree has only leaves, thus it has one or two nodes. Since the local checker is in $\midd{}$ there must be trees with $\Delta>1$. In that case, $v$ should have at least $a_{Degree+1}\geq 1$ neighbors of degree $\Delta+1$, a contradiction. 
\end{proof}

The following results make great use of the grafting operation (Definition~\ref{def:grafting} in Section~\ref{sec:maxdiam}). In our restricted setting, Lemmas~\ref{lem:grafting} and~\ref{lem:pumping_lemma} still hold, where the view of an edge is given by the degrees of its endpoints. Note that in particular, the trees recognized by local checkers in $\midd$ cannot contain any path between edges with the same view. We first use this idea to get rid of the \emph{equal-degree} type.

\begin{claim}
\label{cl:eqdeg}
    There cannot be a path $u_1,u_2,\ldots, v_1,v_2$ with $\deg(u_1)=\deg(u_2)$ and $\deg(v_1)=\deg(v_2)$ in a tree accepted by some $L\in\midd{}$.
\end{claim}

\begin{proof}
Assume that a tree $T$ accepted by $L$ contains such a path. Now graft $T$ in $T$ at $(v_1v_2,v_2v_1)$. By Lemma~\ref{lem:grafting}, the new tree $T'$ is accepted by $L$. Moreover, $T'$ contains two copies of the edge $u_1u_2$, call the second one $u'_1u'_2$. These two edges have the same view, hence one can apply Lemma~\ref{lem:pumping_lemma} to construct trees of linear diameter accepted by $L$, a contradiction.
\end{proof}

Using this result, we can already restrict the study to checkers $L\in\midd$ with $a_{Equal-degree}=b_{Equal-degree}=0$. Indeed, by Claim~\ref{cl:eqdeg}, all accepted trees only use the equal-degree case at most once, hence necessarily $a_{Equal-degree}=0$. Then for every such tree $T$, we cut the equal-degree edge, keep only half of the tree (the one with largest diameter), and fix it in the following way. 
We replace the cut edge by an edge to a leaf, and create $T'$. Observe that $T'$ is accepted by the checker $L'$ which is the same as $L$, except that $b_{leaves}$ is increased by 1, and $b_{Equal-degree}=0$. It is thus sufficient to prove Theorem~\ref{thm:classification-myopic} for $L'$.

From now on, we consider only checkers in $\midd$ with $a_{Equal-degree}=b_{Equal-degree}=0$.
\begin{claim}
    For any local checker in $\midd{}$, $b_{Degree-1}\geq 1$ and $b_{Degree+1}\geq 1$. 
\end{claim}

\begin{proof}
    Since the checker is in $\midd{}$, there must be at least two non-leaf nodes in at least one accepted tree. 
    Consider two such non-leaf nodes $u$ and $v$, and w.l.o.g. take them to be adjacent. 
    Since in our restricted setting, we only allow two adjacent non-leaf nodes to have a degree difference of exactly 1, it must be that, up to symmetry, $v$ is categorizing $u$ as Degree+1, and $u$ is categorizing $v$ as Degree-1. Hence, we must allow at least one neighbor of each type.     
\end{proof}

    %In the rest, we illustrate trees by ordering nodes in increasing degree, except leaves. Note that consecutive layers. 
    %Let a zigzag be four nodes with degree $d,d',d,d'$ on a path. \textcolor{red}{Def to be changed?} 
    %It makes a zigzag in the picture. See picture to do.

Let a \emph{zigzag} in a tree be six nodes $u_1,u_2,u_3$ and $v_1,v_2,v_3$, such that there exists a path in which they appear in this order, and $\deg(u_1)=\deg(u_3)=\deg(u_2)-1$ and $\deg(v_1)=\deg(v_3)=\deg(v_2)+1$. See Figure~\ref{fig:zigzag}. 

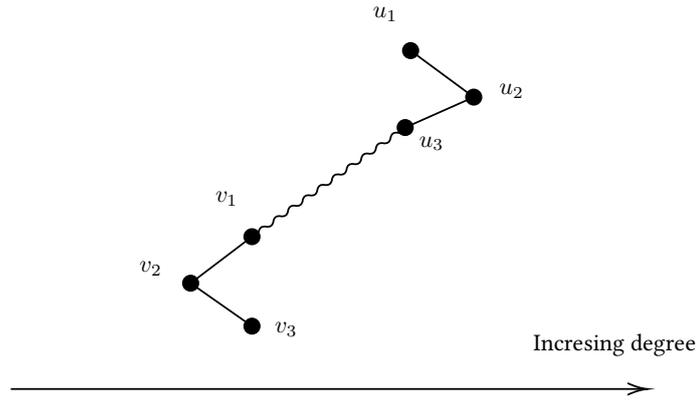
\begin{figure}[!h]
    \centering
    \scalebox{0.9}{
    \tikzset{every picture/.style={line width=0.75pt}} %set default line width to 0.75pt        

\begin{tikzpicture}[x=0.75pt,y=0.75pt,yscale=-1,xscale=1]

%Shape: Circle [id:dp3933205737914821] 
\draw  [fill={rgb, 255:red, 0; green, 0; blue, 0 }  ,fill opacity=1 ] (339,87.25) .. controls (339,84.9) and (340.9,83) .. (343.25,83) .. controls (345.6,83) and (347.5,84.9) .. (347.5,87.25) .. controls (347.5,89.6) and (345.6,91.5) .. (343.25,91.5) .. controls (340.9,91.5) and (339,89.6) .. (339,87.25) -- cycle ;
%Shape: Circle [id:dp7824240878395652] 
\draw  [fill={rgb, 255:red, 0; green, 0; blue, 0 }  ,fill opacity=1 ] (374,113.25) .. controls (374,110.9) and (375.9,109) .. (378.25,109) .. controls (380.6,109) and (382.5,110.9) .. (382.5,113.25) .. controls (382.5,115.6) and (380.6,117.5) .. (378.25,117.5) .. controls (375.9,117.5) and (374,115.6) .. (374,113.25) -- cycle ;
%Shape: Circle [id:dp7859026792506623] 
\draw  [fill={rgb, 255:red, 0; green, 0; blue, 0 }  ,fill opacity=1 ] (336,130.25) .. controls (336,127.9) and (337.9,126) .. (340.25,126) .. controls (342.6,126) and (344.5,127.9) .. (344.5,130.25) .. controls (344.5,132.6) and (342.6,134.5) .. (340.25,134.5) .. controls (337.9,134.5) and (336,132.6) .. (336,130.25) -- cycle ;
%Shape: Circle [id:dp3862395122277267] 
\draw  [fill={rgb, 255:red, 0; green, 0; blue, 0 }  ,fill opacity=1 ] (251,191.25) .. controls (251,188.9) and (252.9,187) .. (255.25,187) .. controls (257.6,187) and (259.5,188.9) .. (259.5,191.25) .. controls (259.5,193.6) and (257.6,195.5) .. (255.25,195.5) .. controls (252.9,195.5) and (251,193.6) .. (251,191.25) -- cycle ;
%Shape: Circle [id:dp4571688610438298] 
\draw  [fill={rgb, 255:red, 0; green, 0; blue, 0 }  ,fill opacity=1 ] (217,217.25) .. controls (217,214.9) and (218.9,213) .. (221.25,213) .. controls (223.6,213) and (225.5,214.9) .. (225.5,217.25) .. controls (225.5,219.6) and (223.6,221.5) .. (221.25,221.5) .. controls (218.9,221.5) and (217,219.6) .. (217,217.25) -- cycle ;
%Shape: Circle [id:dp9903874317341745] 
\draw  [fill={rgb, 255:red, 0; green, 0; blue, 0 }  ,fill opacity=1 ] (251,241.25) .. controls (251,238.9) and (252.9,237) .. (255.25,237) .. controls (257.6,237) and (259.5,238.9) .. (259.5,241.25) .. controls (259.5,243.6) and (257.6,245.5) .. (255.25,245.5) .. controls (252.9,245.5) and (251,243.6) .. (251,241.25) -- cycle ;
%Straight Lines [id:da03912011181279951] 
\draw    (343.25,87.25) -- (378.25,113.25) ;
%Straight Lines [id:da5699205697029448] 
\draw    (340.25,130.25) -- (378.25,113.25) ;
%Straight Lines [id:da39746400766376855] 
\draw    (221.25,217.25) -- (255.25,191.25) ;
%Straight Lines [id:da40541634324119835] 
\draw    (221.25,217.25) -- (255.25,241.25) ;
%Straight Lines [id:da10973254114206843] 
\draw    (340.25,130.25) .. controls (339.87,132.58) and (338.52,133.55) .. (336.19,133.17) .. controls (333.86,132.78) and (332.51,133.75) .. (332.13,136.08) .. controls (331.74,138.41) and (330.39,139.38) .. (328.06,139) .. controls (325.73,138.61) and (324.38,139.58) .. (324,141.91) .. controls (323.62,144.24) and (322.27,145.21) .. (319.94,144.83) .. controls (317.61,144.44) and (316.26,145.41) .. (315.88,147.74) .. controls (315.49,150.07) and (314.14,151.04) .. (311.81,150.66) .. controls (309.48,150.27) and (308.13,151.24) .. (307.75,153.57) .. controls (307.37,155.9) and (306.02,156.87) .. (303.69,156.49) .. controls (301.36,156.1) and (300.01,157.07) .. (299.63,159.4) .. controls (299.25,161.73) and (297.9,162.7) .. (295.57,162.32) .. controls (293.24,161.93) and (291.89,162.9) .. (291.5,165.23) .. controls (291.12,167.56) and (289.77,168.53) .. (287.44,168.15) .. controls (285.11,167.76) and (283.76,168.73) .. (283.38,171.06) .. controls (283,173.39) and (281.65,174.36) .. (279.32,173.98) .. controls (276.99,173.59) and (275.64,174.56) .. (275.25,176.89) .. controls (274.87,179.22) and (273.52,180.19) .. (271.19,179.81) .. controls (268.86,179.42) and (267.51,180.39) .. (267.13,182.72) .. controls (266.75,185.05) and (265.4,186.02) .. (263.07,185.64) .. controls (260.74,185.25) and (259.39,186.22) .. (259.01,188.55) -- (255.25,191.25) -- (255.25,191.25) ;
%Straight Lines [id:da8488065523494417] 
\draw    (121.5,276.37) -- (472.5,276.37) ;
\draw [shift={(474.5,276.37)}, rotate = 180] [color={rgb, 255:red, 0; green, 0; blue, 0 }  ][line width=0.75]    (10.93,-3.29) .. controls (6.95,-1.4) and (3.31,-0.3) .. (0,0) .. controls (3.31,0.3) and (6.95,1.4) .. (10.93,3.29)   ;

% Text Node
\draw (321,61.4) node [anchor=north west][inner sep=0.75pt]    {$u_{1}$};
% Text Node
\draw (391,104.4) node [anchor=north west][inner sep=0.75pt]    {$u_{2}$};
% Text Node
\draw (346.5,133.65) node [anchor=north west][inner sep=0.75pt]    {$u_{3}$};
% Text Node
\draw (233.5,164.65) node [anchor=north west][inner sep=0.75pt]    {$v_{1}$};
% Text Node
\draw (191.5,203.65) node [anchor=north west][inner sep=0.75pt]    {$v_{2}$};
% Text Node
\draw (266.5,237.65) node [anchor=north west][inner sep=0.75pt]    {$v_{3}$};
% Text Node
\draw (410,244) node [anchor=north west][inner sep=0.75pt]   [align=left] {Incresing degree};

\end{tikzpicture}}
    \caption{A zigzag in a (not fully depicted) tree, where the nodes are ordered by increasing degree from left to right. The curvy line represents a path in the tree. On this picture the degree of the $v_i$ is smaller than the degrees of the $u_i$, but this is not necessary. Note that the path in the middle could itself contain zigzags.}
    \label{fig:zigzag}
\end{figure}

Working similarly as Claim~\ref{cl:eqdeg}, we prove that zigzags cannot appear.

\begin{claim}\label{clm:zigzag}
    Let $L$ be a checker in $\midd{}$. The trees accepted by $L$ have no zigzag. 
\end{claim}

\begin{proof}
    By contradiction, suppose we have a checker in $\midd{}$ accepting a tree $T$ with a zigzag $u_1,u_2,u_3, v_1, v_2,v_3$ as in the definition. 
    By Lemma~\ref{lem:grafting}, the tree $T'$ obtained by grafting $T_{v_1}$ in $T$ at $(v_2v_3,v_2v_1)$ is still accepted by $L$.
    
    The new tree $T'$ contains the original vertices $u_1,u_2,u_3$ from $T$ and also a copy of those from the grafted copy of $T_{v_1}$, that we call $u_1',u_2',u_3'$. In particular, $T'$ contains a path $u_1,u_2,\ldots,u'_3,u'_2$, where $u_1u_2$ has the same view as $u'_3u'_2$. We can thus apply Lemma~\ref{lem:pumping_lemma} to construct trees of linear diameter accepted by $L$ and reach a contradiction.
\end{proof}

Forbidding zigzags basically proves that paths in accepted trees can be split in two parts, each monotonous in degree (up to vertices of degree at most $a_{Degree-1}$). This allows us to restrict the values of $a_{leaves}$ and $b_{Degree+1}$ without changing much the maximum diameter.

\begin{claim}\label{clm:zero-leaf-only}
    Let $L$ be a myopic local checker with $a_{leaves} >0$ and maximum diameter $D(n)$. Let $L'$ be the same checker except that  $a_{leaves} =0$. The maximum diameter function $D'(n)$ for $L'$ is at most $D(a_{leaves}\cdot n)$.
\end{claim}

\begin{proof}
    Let $T'$ be an $n$-vertex tree of diameter $D'(n)$ accepted by $L'$. By adding $a_{leaves}$ pending leaves to each non-leaf vertex, we get a tree $T$ accepted by $L$, with at most $a_{leaves}\cdot n$ nodes and whose diameter is still $D'$. In particular, we get $D'(n)\leqslant D(a_{leaves}\cdot n)$.
\end{proof}

Since every tree accepted by $L$ is also accepted by $L'$, we also get $D\leqslant D'$. Therefore, given the diameter functions $D$ we target, we have $D'(n)=\Theta(D(n))$, hence we can assume that $a_{leaves}=0$ in the context of the theorem.  Using slightly more involved arguments, we get that we may also assume that $b_{Degree+1}=1$.

\begin{claim}\label{clm:one-forward-only}
    Let $L$ be a myopic local checker with $b_{Degree+1} >1$ and maximum diameter $D$. Let $L'$ be the same checker except that  $b_{Degree+1} =1$. The maximum diameter function $D'$ for $L'$ is at least $D/2$.
\end{claim}

\begin{proof}
    Consider an $n$-vertex tree $T$ accepted by $L$. And take a path $P$ that is maximum, that is, its length is $D(n)$. Because of Claim~\ref{clm:zigzag}, this path either is monotone in terms of the degree of the nodes visited, or the degree sequence changes slope at most once. 
    If it is monotone or is minimum on the endpoints, then for every node having more than one neighbor of larger degree, we prune all such neighbors except the ones used by the path $P$. This is possible because $a_{Degree+1}=0$ by Claim~\ref{clm:aD+1=0}. This tree is accepted by $L'$ and it has the same diameter as $T'$, and at most the same number of nodes. 
    If the degree sequence is maximum on both endpoints, then we do the same operation, leaving exactly one node with two neighbors of larger degree. Then, we take this node and prune its shortest "forward branch". The diameter has at most halved, while the number of nodes has not increased. 

    In both cases, we obtained a tree accepted by $L'$ with $n'\geqslant D(n)$ vertices and diameter at least $D(n)/2\geqslant D(n')/2$. In particular, we get an infinite sequence of trees accepted by $L'$ with diameter at least $D/2$, which proves the claim.
\end{proof}

%%%%%%%%%%%%%%%%%%%%%%%%%%%%%%%%%%%%%%%%%%%
\subsection{Establishing the classification}
\label{subsec:classification-proof}
%%%%%%%%%%%%%%%%%%%%%%%%%%%%%%%%%%%%%%%%%%%

The claims above show that we may only prove Theorem~\ref{thm:classification-myopic} for degree-myopic local checkers in $\midd$ satisfying:
\begin{itemize}
    \item $a_{leaves}=0$ and $b_{leaves}\geq 1$
    \item $b_{Degree-1}\geq 1$
    \item $a_{Equal-degree}=b_{Equal-degree}=0$
    \item $a_{Degree+1}=0$ and $b_{Degree+1}\geq 1$
\end{itemize}

We now proceed to the proof of Theorem~\ref{thm:classification-myopic}, by distinguishing some cases based on the values of $b_{leaves}$ and $a_{Degree-1}$. In each case, we provide a lower bound on the maximum diameter by constructing trees with a specific shape that are accepted. Then we provide a complementary upper bound by proving that the same kind of structure necessarily appears in the accepted trees.

\begin{claim}
    If $L\in\midd{}$ satisfies $b_{leaves}=\infty$ and $a_{Degree-1}\leq 1$, then its maximum diameter is $\Theta(\sqrt{n})$.
\end{claim}

\begin{proof}
Let $L\in\midd{}$ such that $b_{leaves}=\infty$ and $a_{Degree-1}\leq 1$. For every integer $i$, consider the caterpillar $T_i$ whose backbone has $i$ nodes, and its degree sequence is $1,2,\cdots, i$. Observe that $T_i$ is accepted by $L$ and moreover, $T_i$ has diameter $\Theta(i)$ and $\Theta(i^2)$ vertices. Therefore, $L$ has maximum diameter at least $\Omega(\sqrt{n})$.

To conclude, we show that every $n$-vertex tree $T$ accepted by $L$ has diameter at most $O(\sqrt{n})$. Take a maximum path $P$ in $T$, of length $D$. Since there is no zigzag in the tree, the degree sequence of any maximum path changes slope at most once, so up to halving $P$ (similarly to Claim~\ref{clm:one-forward-only}), one can assume that $P$ has a monotone degree sequence, from a vertex of degree $a_{Degree-1}$ to a vertex of degree $a_{Degree-1}+D/2$. In particular, there are at least $\sum_{i=0}^{D/2}(a_{Degree-1}+i)=\Theta(D^2)$ vertices at distance at most one from~$P$, so $D=O(\sqrt{n})$.
\end{proof}

\begin{claim}
    If $L\in\midd{}$ satisfies $b_{leaves}=\infty$ and $a_{Degree-1}\geq 2$, then its maximum diameter is $\Theta(\log{n})$.
\end{claim}

\begin{proof}
     Let $L\in\midd{}$ such that $b_{leaves}=\infty$ and $a_{Degree-1}\geq 2$. For every integer $i$, denote by $T_i$ the complete $a_{Degree-1}$-ary tree of height $i$ (where each internal node has $a_{Degree-1}$ children). Construct $T'_i$ from $T_i$ by attaching $j$ leaves to every vertex at distance $j+1$ from the leaves. Observe that $T'_i$ is accepted by $L$, because (1) all the leaves and their parents (which have degree $a_{Degree-1}$) accept, and (2) all the other nodes have $a_{Degree-1}$ children (whose degree is precisely one less). 
     
     The tree $T'_i$ has diameter $\Theta(i)$ and a number of nodes $\Theta(\sum_{j=0}^i (j+1)\cdot a_{Degree-1}^{i-j})=\Theta(a_{Degree-1}^i)$ vertices, hence $L$ has maximum diameter $\Omega(\log n)$.

     Let $T$ be an $n$-vertex tree of diameter $D$ accepted by $L$. Orient its edges from $u$ to $v$ if $\deg(v)>\deg(u)$. One can easily prove by induction on $k$ that if $T$ contains a vertex $u$ of degree $k+a_{Degree-1}$, then the subtree rooted at $u$ contains at least $a_{Degree-1}^k$ nodes (it actually contains a copy of $T_k$). Since any maximal path of $T$ contains a vertex of degree at least $D/2$, $T$ has at least $a_{Degree-1}^{D/2-1}$ nodes so $D=O(\log n)$.
\end{proof}

\begin{claim}
    If $L\in\midd{}$ satisfies $b_{leaves}< \infty$,  then its maximum diameter is $\Theta(\log n /\log \log n)$.  
\end{claim}

\begin{proof}
Let $L\in\midd{}$ with $b_{leaves}<\infty$. 
Let us prove that necessarily $b_{Degree-1}=\infty$. Suppose that this is not the case, then $L$ accepts only trees of maximum degree $\Delta$ at most $b_{leaves}+b_{Degree-1}+1$. If $T$ is accepted by $L$, it contains no zigzag, hence any path in $T$ has length at most $2\Delta$. Hence, $L$ has maximum diameter $O(1)$, a contradiction with the definition of $\midd{}$. From now on, assume that $b_{Degree-1}=\infty$.

For every integer $i\geq a_{Degree-1}-1$, we define inductively a rooted tree $T_i$. Let $T_{a_{Degree-1}-1}$ be the subdivided star whose root has degree $a_{Degree-1}-1$ and each branch is subdivided once. Then, construct $T_{i+1}$ by attaching $i+1$ copies of $T_i$ to a new root. Note that $T_i$ has diameter $\Theta(i)$ and $\Theta(i!)$ nodes. 
We check that $L$ accepts all the trees $T_i$, and has maximum diameter $\Omega(\log n/\log\log n)$ (the asymptotic inverse function of factorial).

Let $T$ be an $n$-vertex tree of diameter $D$ accepted by $L$, so in particular $T$ contains a vertex of degree at least $D/2$. Orient each edge of $T$ towards its endpoint of largest degree. Again, one can easily show by induction on $k$ that each vertex $u$ with $k+b_{leaves}$ children has at least $k!$ descendants. In particular, $T$ contains at least $(D/2 -b_{leaves})!$ nodes, hence $D=O(\log n/\log\log n)$.
\end{proof}

\section{Towards graphs with cycles}
\label{sec:cycles}

Up to now, we have focused on trees, because they are central for LCLs, and make sense for our second motivation.  
A natural question at that point is what happens if we go beyond trees. We leave this for further work. 
We just prove now that if there is a cycle in the graph, but no node can see it (\emph{i.e}, the girth is large in comparison with the checkability radius) the maximum diameter is linear. 
Basically, we show that we can do pumping via a crossing argument. 

Consider a local checker that accepts some graph with a cycle. We claim that it also accepts arbitrarily large graphs of linear diameter, relying on the so-called \emph{duplication} operation. Given a colored graph $G$ and an edge $uv\in E(G)$, the duplication of $G$ along $uv$ is the colored graph $G^{uv}$ obtained by taking two copies of $G$, deleting the two copies $u_1v_1$ and $u_2v_2$ of $uv$ and adding $u_1v_2$ and $u_2v_1$. It is easy to see that if $u$ and $v$ are at distance at least $2d$ from each other in $G\setminus uv$, then $u_1$ and $u_2$ (resp. $v_1$ and $v_2$) both have the same view at distance $d$ in $G^{uv}$ as $u$ (resp. $v$) in $G$. In particular, we get the following.
\begin{lemma}
\label{lem:cycle_pump}
    Let $c,d$ be integers, $G$ be a $c$-colored graph containing two adjacent vertices $u,v$ at distance at least $2d$ in $G\setminus uv$. Then every $L\in\mathcal{L}_{c,d}$ accepting $G$ also accepts $G^{uv}$.
\end{lemma}

The following structural result shows that this operation creates long paths, which is the main ingredient for constructing graphs of linear diameter.

\begin{lemma}
\label{lem:cycle_diam}
    Let $uv$ be an edge of a graph $G$. Denote by $u'v'$ the edge between copies of $u$ and $v$ in $G^{uv}$. Then  $d_{G^{uv}\setminus u'v'}(u',v')=2d_{G\setminus uv}(u,v)+1$.
\end{lemma}

Using these two results, we can conclude that iterating this duplication operation starting from a nice enough graph $G$ yields an infinite sequence of graphs of linear diameter, all accepted by every local checker accepting $G$.

\begin{theorem}
    Let $c,d$ be integers and $L\in \mathcal{L}_{c,d}$ accepting a $c$-colored graph $G$ containing two adjacent vertices $u,v$ at distance at least $2d$ in $G\setminus uv$. Then $L$ has maximum diameter $\Theta(n)$.
\end{theorem}

\begin{proof}
    Consider the sequence defined by $G_0=G$, $u_0=u$, $v_0=v$ and for every $i$, $G_{i+1}=G_i^{u_iv_i}$, and $u_{i+1}v_{i+1}$ is an edge in $G_{i+1}$ between the copies of $u_i$ and $v_i$. By Lemma~\ref{lem:cycle_pump}, all the graphs $G_i$ are accepted by $L$. Moreover, $|G_i|=2^i\cdot |G|$ and by Lemma~\ref{lem:cycle_diam}, the distance between $u_i$ and $v_i$ in $G_i\setminus u_iv_i$ is at least $2^i\cdot d$. Therefore, there are two vertices in $G_i$ at distance at least $2^{i-1}\cdot d$ hence $\diam(G_i)=\Theta(|G_i|)$, which concludes.
\end{proof}

\begin{comment}
    Let us prove that using this lemma we can construct a graph of linear diameter accepted if there is at least one graph with a cycle that is accepted.

Let $G$ be a graph with a cycle accepted of minimum size $c$ and let $C$ be a cycle of minimum length of $G$. Let us construct a family of graphs $G_1,\ldots,G_r$ such that the size of $G_i$ is $|V(G)| \cdot i$ and there exists an edge $u_iv_i$ in $G_i$, a cycle $C_i$ of length at least $i+2$ containing $u_iv_i$ such that $d(u_i,v_i)$ in $G_i \setminus u_iv_i$ is $|C_i|-1$. Then we bootstrap $G_i$ with $G_i$ to double the length of the cycle.
\textcolor{red}{To write properly using colors !}
\textcolor{red}{We apply a logarithmic number of times to get a linear bound on the size of the cycle (to write properly)}

\end{comment}

\end{document}